%% file: mmWave_MU_Revised_Final.tex
\documentclass[onecolumn,draftclsnofoot,12pt]{IEEEtran}
\input{input.tex}

\newcommand{\sref}[1]{{Section}~\ref{#1}}
\usepackage{epstopdf}
\usepackage{enumerate}
\usepackage{algorithmicx}
\usepackage{algorithm}
\usepackage{amsmath}
\usepackage[noend]{algpseudocode}
\usepackage{float}
\usepackage{color}
\usepackage{makeidx}
\usepackage{bbm}

\allowdisplaybreaks
\def\j{\mathrm{j}}
\begin{document}
\title{ Limited Feedback Hybrid Precoding for Multi-User Millimeter Wave Systems}
\author{Ahmed Alkhateeb, Geert Leus, and Robert W. Heath Jr.\thanks{Ahmed Alkhateeb and Robert W. Heath Jr. are with The University of Texas at Austin (Email: aalkhateeb, rheath,@utexas.edu). Geert Leus is with Delft University of Technology (Email: g.j.t.leus@tudelft.nl).} \thanks{This work is supported in part by the National Science Foundation under Grant No. 1218338 and 1319556, and by a gift from Huawei Technologies, Inc.}}
\maketitle

\begin{abstract}
Antenna arrays will be an important ingredient in millimeter wave (mmWave) cellular systems. A natural application of antenna arrays is simultaneous transmission to multiple users. Unfortunately, the hardware constraints in mmWave systems make it difficult to apply conventional lower frequency multiuser MIMO precoding techniques at mmWave. This paper develops low-complexity hybrid analog/digital precoding for downlink multiuser mmWave systems. Hybrid precoding involves a combination of analog and digital processing that is inspired by the power consumption of complete radio frequency and mixed signal hardware. The proposed algorithm configures hybrid precoders at the transmitter and analog combiners at multiple receivers with a small training and feedback overhead. The performance of the proposed algorithm is analyzed in the large dimensional regime and in single path channels. When the analog and digital precoding vectors are selected from quantized codebooks, the rate loss due to the joint quantization is characterized and insights are given into the performance of hybrid precoding compared with analog-only beamforming solutions. Analytical and simulation results show that the proposed techniques offer higher sum rates compared with analog-only beamforming solutions, and approach the performance of the unconstrained digital beamforming with relatively small codebooks.
\end{abstract}

\section{Introduction} \label{sec:intro}

The large bandwidths in the mmWave spectrum make mmWave communication desirable for wireless local area networking and also a promising candidate for future cellular systems \cite{Pi2011,Cov_Magazine,Rapp5G,ayach2013spatially,mmWave_Estimation_2013}. Achieving high quality communication links in mmWave systems requires employing large antenna arrays at both the access point or base station (BS) and the mobile stations (MS's) \cite{murdock201238,ayach2013spatially,Antenna2}. For efficient system performance, each BS needs to simultaneously serve a number of MS's. Multiplexing different data streams to different users requires some form of precoding be applied to generate the transmitted signal at the BS. In conventional lower frequency systems, this precoding was commonly done in the baseband to have a better control over the entries of the precoding matrix. Unfortunately, the high cost and power consumption of mixed signal components make fully digital baseband precoding unlikely with current semiconductor technologies \cite{ayach2013spatially,singh2009multi}. Further, the design of the precoding matrices is usually based on complete channel state information, which is difficult to achieve in mmWave systems due to (i) the large number of antennas which would require a huge training overhead and (ii) the small signal-to-noise ratio (SNR) before beamforming. Therefore, new multi-user precoding algorithms that (i) respect the mmWave hardware constraints and (ii) require much less complexity need to be developed for mmWave systems.

In single-user mmWave systems, analog beamforming, which controls the phase of the signal transmitted at each antenna via a network of analog phase shifters and is implemented in the radio frequency (RF) domain, was proposed instead of the baseband solutions \cite{Wang1,Multilevel,Tsang,sayeed2007maximizing,brady2013beamspace}. This was also adopted in commercial indoor mmWave communication standards like IEEE 802.11ad \cite{.11ad} and IEEE 802.15.3c \cite{IEEE3c}. In \cite{Wang1,Multilevel}, adaptive beamforming algorithms and multi-resolution codebooks were developed by which the transmitter and receiver jointly  design their analog beamforming vectors. In \cite{Tsang}, unique signatures are assigned to the different training beamforming vectors and used to minimize the training overhead. In \cite{sayeed2007maximizing,brady2013beamspace}, beamspace multi-input multi-output (MIMO) was introduced in which discrete fourier transform (DFT) beamforming vectors are used to direct the transmitted signals towards the subspaces that asymptotically maximize the received signal power with large numbers of antennas. Analog beamformers as in \cite{.11ad,IEEE3c,Wang1,Multilevel,Tsang,sayeed2007maximizing,brady2013beamspace} are subject to additional constraints, for example,
the phase shifters might be digitally controlled and have only quantized phase values and adaptive gain control might not be implemented. These constraints limit the potential of analog-only beamforming solutions relative to baseband precoding, as they limit the ability to make sophisticated processing, for example to manage interference between users.

To multiplex several data streams and perform more accurate beamforming, hybrid precoding was proposed \cite{ayach2013spatially,Multi_beam,alkhateeb}, where the processing is divided between the analog and digital domains. In \cite{ayach2013spatially}, the sparse nature of the mmWave channels was exploited to develop low-complexity hybrid precoding algorithms using the algorithmic concept of basis pursuit assuming the availability of channel knowledge. In \cite{Multi_beam}, low-complexity hybrid beamforming algorithms were proposed for single-user single-stream  MIMO-OFDM systems with the objective of maximizing either the received signal strength or the sum-rate over different sub-carriers. In \cite{alkhateeb}, a hybrid precoding algorithm that requires only partial knowledge about the mmWave channels  was devised. The hybrid precoding algorithms in \cite{ayach2013spatially,Multi_beam,alkhateeb}, though, were designed to obtain either diversity or spatial multiplexing gain from single-user channels, which can support a limited number of streams \cite{Rapp5G}. In multi-user systems, the digital precoding layer of hybrid precoding gives more freedom in designing the precoders, compared with analog-only solutions, which can be exploited to reduce the interference between users. Hence, developing low-complexity hybrid precoding algorithms for multi-user mmWave systems is of special interest.

Pre-precoding processing has been investigated for other systems \cite{Zhang,Venkat,Adhikary}. In \cite{Zhang}, the joint analog-digital precoder design problem was studied for both diversity and spatial multiplexing systems. In \cite{Venkat}, hybrid analog/digital precoding algorithms were developed to minimize the received signal's mean-squared error in the presence of interference when phase shifters with only quantized phases are available.  In \cite{Adhikary}, two-layer beamforming algorithms were proposed to group the users and reduce the channel feedback overhead in massive MIMO systems. The approaches in \cite{Zhang,Venkat,Adhikary}, however, were not designed specifically for mmWave systems as they did not consider the mmWave-related hardware constraints, and did not leverage mmWave channel characteristics to realize low-complexity solutions.

In this paper, we develop a low-complexity yet efficient hybrid analog/digital precoding algorithm for downlink multi-user mmWave systems. The proposed algorithm is general for arbitrary known array geometries, and assumes the availability of only a limited feedback channel between the BS and MS's. The main contributions of the paper can be summarized as follows:
\begin{itemize}
\item{Developing a hybrid precoding/combining algorithm for multi-user mmWave systems. Our model assumes that the MS's employ analog-only combining while the BS performs hybrid analog/digital precoding where the number of RF chains is at least as large as the number of MS's. The proposed algorithm is designed to reduce the training and feedback overhead while achieving performance close to that of unconstrained solutions.}
\item{Analyzing the performance of the proposed algorithm in special cases: (i) when the channels are single-path, and (ii) when the number of transmit and receive antennas are very large, which are relevant for mmWave systems.}
\item{Characterizing the average rate loss due to joint analog and digital codebook quantization, and identifying the cases at which large hybrid precoding gains exist compared with analog-only beamforming solutions.}
\end{itemize}
The proposed algorithm and performance bounds are evaluated by simulations and compared with both analog-only beamforming solutions and digital unconstrained precoding schemes. The results indicate that with a relatively small feedback and training overhead, the proposed hybrid precoding algorithm achieves good performance thanks to the sparse nature of the channel and the large number of antennas used by the BS and MS's.

The rest of the paper is organized as follows. In \sref{sec:Model}, the system and channel models are described. In \sref{sec:Form}, the multi-user hybrid precoding/combing design problem is formulated and the large feedback and training overhead associated with the direct solution is explained. The proposed low-complexity solution is then presented in \sref{sec:MU_Alg}, and analyzed in \sref{sec:Performance} assuming the availability of infinite feedback channels and continuous-angle phase shifters. The rate loss due to quantization and limited feedback channels in then characterized in \sref{sec:Finite}. Simulation results are presented  in \sref{sec:Results} before concluding the paper in \sref{sec:conclusion}.

We use the following notation throughout this paper: $\bA$ is a matrix, $\ba$ is a vector, $a$ is a scalar, and $\cA$ is a set. $\|\bA \|_F$ is the Frobenius norm of $\bA$, whereas $\bA^\mathrm{T}$, $\bA^*$, $\bA^{-1}$, are its transpose, Hermitian, and inverse respectively. $\bI$ is the identity matrix, and $\cN(\bm,\bR)$ is a complex Gaussian random vector with mean $\bm$ and covariance $\bR$. $\mathbbm{1}_{(.)}$ is an indicator function. $\bbE\left[\cdot\right]$ is used to denote expectation.

\section{System Model} \label{sec:Model}

\begin{figure} [t]
\centerline{
\includegraphics[scale=.53]{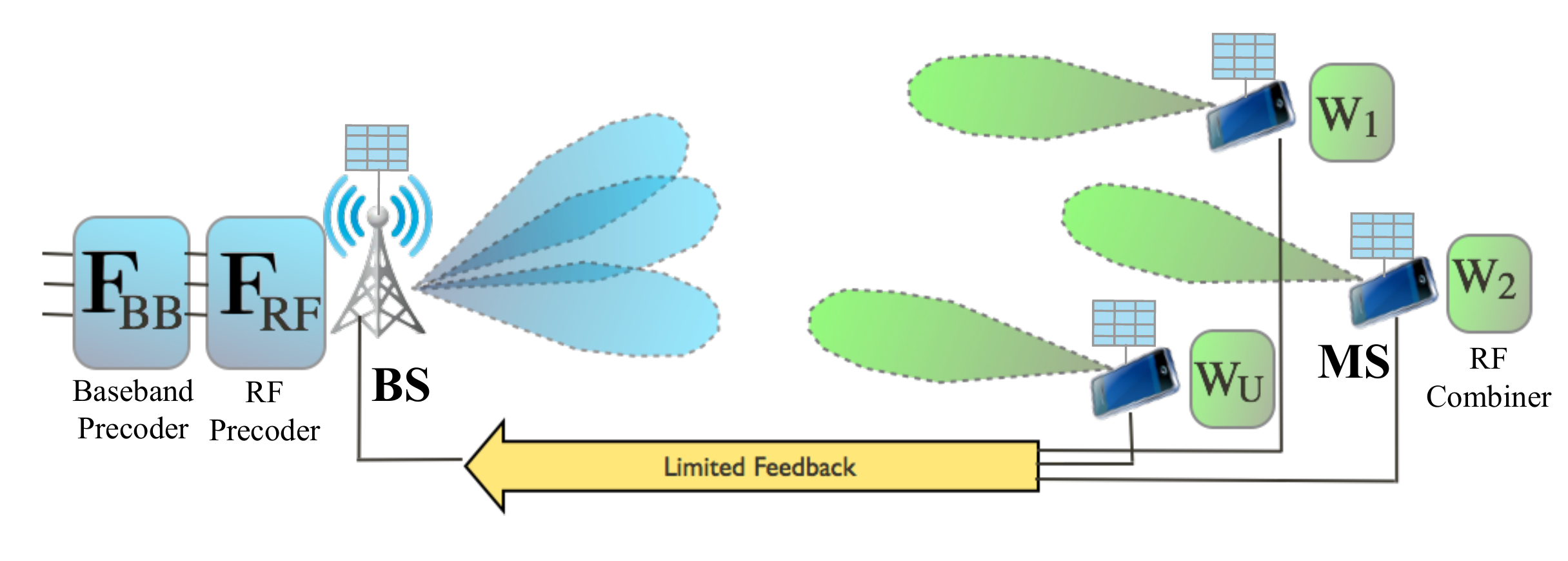}
}
\caption{A multi-user mmWave downlink system model, in which a BS uses hybrid analog/digital precoding and a large antenna array to serve $U$ MSs. Each MS employs analog-only combining and has a limited feedback channel to the BS.}
\label{fig:LimitedFB}
\end{figure}

Consider the multi-user mmWave system shown in \figref{fig:LimitedFB}. A base station  with $N_\mathrm{BS}$ antennas and $N_\mathrm{RF}$ RF chains is assumed to communicate with $U$ mobile stations. Each MS is equipped with $N_\mathrm{MS}$ antennas as depicted in \figref{fig:Hybrid}. We focus on the multi-user beamforming case in which the BS communicates with every MS via \textit{only one stream}. Therefore, the total number of streams $N_\mathrm{S}=U$. Further, we assume that the maximum number of users that can be simultaneously served by the BS equals the number of BS RF chains, i.e., $U \leq N_\mathrm{RF}$. This is motivated by the spatial multiplexing gain of the described multi-user hybrid precoding system, which is limited by $\min\left(N_\mathrm{RF},U\right)$ for $N_\mathrm{BS} > N_\mathrm{RF}$. For simplicity, we will also assume that the BS will use $U$ out of the $N_\mathrm{RF}$ available RF chains to serve the $U$ users.

On the downlink, the BS applies a $U \times U$ baseband precoder $\bF_\mathrm{BB}=\left[\bff_1^\mathrm{BB}, \bff_2^\mathrm{BB}, ..., \bff_U^\mathrm{BB}\right]$ followed by an $N_\mathrm{BS} \times U$ RF precoder, $\bF_\mathrm{RF}=\left[\bff_1^\mathrm{RF}, \bff_2^\mathrm{RF}, ..., \bff_U^\mathrm{RF}\right]$. The sampled transmitted signal is therefore
\begin{equation}
\bx=\bF_\mathrm{RF} \bF_\mathrm{BB} \bs,
\label{eq:signal_transmitted}
\end{equation}
where $\bs=[s_1, s_2, ..., s_U]^\mathrm{T}$ is the $U \times 1$ vector of transmitted symbols, such that $\bbE\left[\bs\bs^*\right] = \frac{P}{U}\bI_U$, and $P$ is the average total transmitted power. We assume equal power allocation among different users' streams. Since $\bF_\mathrm{RF}$ is implemented using analog phase shifters, its entries are of constant modulus. We normalize these entries to satisfy $\left|\left[\bF_\mathrm{RF}\right]_{m,n}\right|^2=N_\mathrm{BS}^{-1}$. Further, we assume that the angles of the analog phase shifters are quantized and have a finite set of possible values. With these assumptions, $\left[\bF_\mathrm{RF}\right]_{m,n}= \frac{1}{\sqrt{N_\mathrm{BS}}} e^{\j \phi_{m,n}}$, where $\phi_{m.n}$ is a quantized angle. The angle quantization assumption is discussed in more detail in \sref{sec:Form}. The total power constraint is enforced by normalizing $\bF_\mathrm{BB}$ such that  $\|\bF_\mathrm{RF} \bF_\mathrm{BB}\|_F^2=U$.

\begin{figure} [t]
\centerline{
\includegraphics[width=.8\columnwidth]{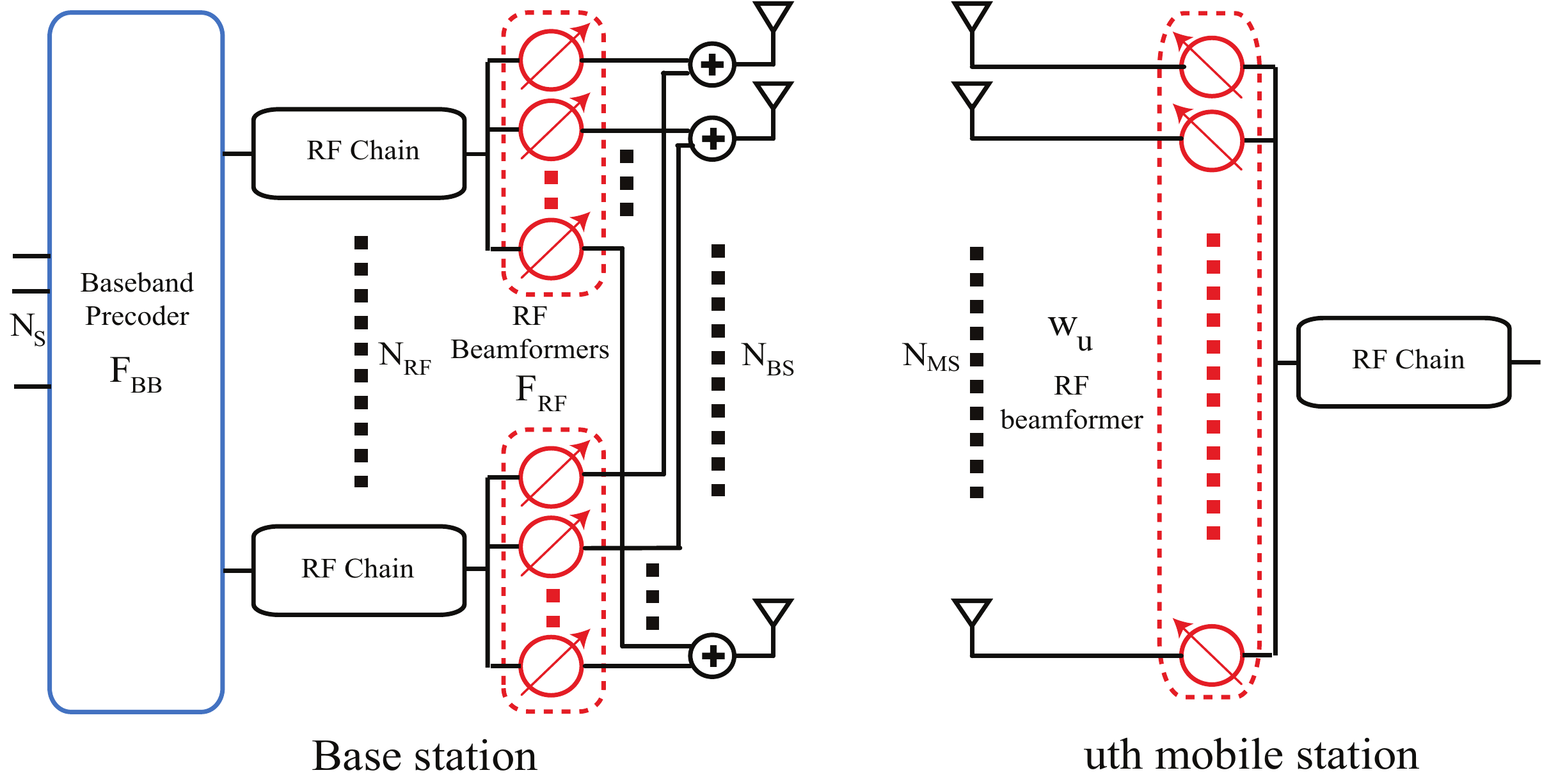}
}
\caption{A BS with hybrid analog/digital architecture communicating with the $u$th MS that employs analog-only combining.}
\label{fig:Hybrid}
\end{figure}


For simplicity, we adopt a narrowband block-fading channel model as in \cite{Xia2008,ayach2013spatially,brady2013beamspace,mmWave_Estimation_2013} in which the $u$th MS observes the received signal as
\begin{equation}
\br_{u}=\bH_{u}\sum_{n=1}^{U}{\bF_\mathrm{RF} \bff^\mathrm{BB}_{n} s_{n}} + \bn_{u},
\label{eq:received_signal}
\end{equation}
where $\bH_{u}$ is the $N_\mathrm{MS} \times N_\mathrm{BS}$ matrix that represents the mmWave channel between the BS and the $u$th MS, and $\bn_{u} \sim \cN (\boldsymbol{0}, \sigma^2 \bI )$ is the Gaussian noise corrupting the received signal.

At the $u$th MS, the RF combiner $\bw_u$ is used to process the received signal $\br_u$:
\begin{equation}
y_u= \bw_u^*  \bH_{u}  \sum_{n=1}^{U}{\bF_\mathrm{RF} \bff^\mathrm{BB}_{n} s_{n}} + \bw_u^* \bn_u,
\label{eq:combined_signal}
\end{equation}
where $\bw_u$ has similar constraints as the RF precoders, i.e., the constant modulus and quantized angles constraints. In this work, we assume that only analog (RF) beamforming is used at the MS's as they will likely need cheaper hardware with lower power consumption.

MmWave channels are expected to have limited scattering~\cite{Ben-Dor,Rapp5G,raghavan2009multi}. To incorporate this effect, we adopt a geometric channel model with $L_{u}$ scatterers for the channel of user $u$. Each scatterer is assumed to contribute a single propagation path between the BS and MS~\cite{ayach2013spatially,raghavan2009multi}. The adopted geometrical channel model can also be transformed into the virtual channel model \cite{Virtual}. The virtual channel model simplifies the generalization for larger angle spreads by incorporating spatial spreading functions as will be briefly discussed in \sref{subsec:Perfect_large}. Under this model, the channel $\bH_{u}$ can be expressed as
\begin{align}
\bH_u =\sqrt{\frac{N_\mathrm{BS} N_\mathrm{MS}} { L_{u}}} \sum_{\ell=1}^{L_{u}}\alpha_{u,\ell} \ba_\mathrm{MS}\left(\theta_{u,\ell}\right) \ba^*_\mathrm{BS} \left(\phi_{u,\ell} \right),
\label{eq:channel_model}
\end{align}
where $\alpha_{u,\ell}$ is the complex gain of the $\ell^\mathrm{th}$ path, including the path-loss, with $\mathbb{E}\left[|\alpha_{u,\ell}|^2\right]=\bar{\alpha}$. The variables $\theta_{u,\ell}$, and $\phi_{u,\ell} \in [0, 2\pi]$ are the $\ell^\mathrm{th}$ path's angles of arrival and departure (AoAs/AoDs) respectively. Finally, $\ba_\mathrm{BS}\left(\phi_{u,\ell}\right)$ and $\ba_\mathrm{MS}\left(\theta_{u,\ell}\right)$ are the antenna array response vectors of the BS and $u$th MS respectively. The BS and each MS are assumed to know the geometry of their antenna arrays. While the algorithms and results developed in the paper can be applied to arbitrary antenna arrays, we use uniform planar arrays (UPAs) and uniform linear arrays (ULAs) in the simulations of \sref{sec:Results}. If a ULA is assumed, $\ba_\mathrm{BS}\left(\phi_{u,\ell}\right)$ can be defined as
\begin{align}
\begin{split}
\ba_\mathrm{BS}\left(\phi_{u,\ell}\right) =  \frac{1}{\sqrt{N_\mathrm{BS}}} \left[ 1, e^{\j{\frac{2\pi}{\lambda}}d\sin\ \left( \phi_{u,\ell} \right)}, ..., e^{\j\left(N_\mathrm{BS} -1\right){\frac{2\pi}{\lambda}}d\sin \left( \phi_{u,\ell} \right)} \right]^\mathrm{T},  \label{eq:steering}
\end{split}
\end{align}
where $\lambda$ is the signal wavelength, and $d$ is the distance between antenna elements. The array response vectors at the MS, $\ba_\mathrm{MS}\left(\theta_{u,\ell}\right)$, can be written in a similar fashion.

\section{Problem Formulation} \label{sec:Form}
Our main objective is to efficiently design the analog (RF) and digital (baseband) precoders at the BS and the analog combiners at the MS's to maximize the sum-rate of the system.

Given the received signal at the $u$th MS in \eqref{eq:received_signal} which is then processed using the RF combiner $\bw_u$, the achievable rate of user $u$ is
\begin{equation}
R_u=\log_2\left(1+\frac{\frac{P}{U} \left|\bw_u^* \bH_u \bF_\mathrm{RF} \bff_u^\mathrm{BB}\right|^2}{\frac{P}{U} \sum_{n\neq u} { \left|\bw_u^* \bH_u \bF_\mathrm{RF} \bff_n^\mathrm{BB}\right|^2 + \sigma^2} }\right). \label{eq:Rate}
\end{equation}
The sum-rate of the system is then  $R_\mathrm{sum}=\sum_{u=1}^U R_u$.

Due to the constraints on the RF hardware, such as the availability of only quantized angles for the RF phase shifters, the analog beamforming/combining vectors can take only certain values. Hence, these vectors need to be selected from finite-size codebooks. There are different models for the RF beamforming codebooks, two possible examples are
\begin{enumerate}
\item{\textbf{General quantized beamforming codebooks} Here, the codebooks are designed to satisfy some particular properties, e.g., maximizing the minimum distance between the codebook vectors as in Grassmannian codebooks. These codebooks are usually designed for rich channels and, therefore, attempt a uniform quantization on the space of beamforming vectors. These codebooks were commonly used in traditional MIMO systems.}
\item{\textbf{Beamsteering codebooks} The beamforming vectors, here, are spatial matched filters for the single-path channels. As a result, they have the same form of the array response vector and can be parameterized by a simple angle. Let $\cF$ represents the RF beamforming codebook, with cardinality $\left|\cF\right|=N_\mathrm{Q}$. Then, in the case of beamsteering codebooks, $\cF$ consists of the vectors $\ba_\mathrm{BS}\left(\frac{2 \pi k_\mathrm{Q}}{N_\mathrm{Q}}\right)$, for the variable $k_\mathrm{Q}$ taking the values $0, 1, 2$, and $N_\mathrm{Q}-1$. The RF combining vectors codebook $\cW$ can be similarly defined.}
\end{enumerate}

Motivated by the good performance of single-user hybrid precoding algorithms \cite{ayach2013spatially,mmWave_Estimation_2013} which relied on RF beamsteering vectors, and by the relatively small size of these codebooks which depend on single parameter quantization, we will adopt the beamsteering codebooks for the analog beamforming vectors. While the problem formulation and proposed algorithm in this paper are general for any codebook, the performance evaluation of the proposed algorithm done in Sections \ref{sec:Performance}-\ref{sec:Finite} depends on the selected codebook. For future work, it is of interest to evaluate the performance of the proposed hybrid precoding algorithm with other RF beamforming codebooks.

If the system sum-rate is adopted as a performance metric, the precoding design problem is then to find $\bF_\mathrm{RF}^\star$, $\left\{\bff_u^{ \star \mathrm{BB}}\right\}_{u=1}^U$ and $\left\{\bw_u^\star\right\}_{u=1}^U$ that solve
\begin{align}
\begin{split}
\left\{\bF_\mathrm{RF}^\star, \left\{\bff_u^{ \star \mathrm{BB}}\right\}_{u=1}^U, \left\{\bw_u^\star\right\}_{u=1}^U\right\} = & \arg\max {\sum_{u=1}^{U} \log_2\left(1+\frac{\frac{P}{U} \left|\bw_u^* \bH_u \bF_\mathrm{RF} \bff_u^\mathrm{BB}\right|^2}{\frac{P}{U} \sum_{n\neq u} { \left|\bw_u^* \bH_u \bF_\mathrm{RF} \bff_n^\mathrm{BB}\right|^2 + \sigma^2} }\right)} \\
& \mathrm{s.t.} \hspace{23 pt}\left[\bF_\mathrm{RF}\right]_{:,u} \in \cF, u=1,2,..., U, \\
& \hspace{40pt}  \bw_u \in \cW, u=1,2,..., U, \\
& \hspace{40pt}  \|\bF_\mathrm{RF} \left[\bff_1^\mathrm{BB}, \bff_2^\mathrm{BB}, ..., \bff_U^\mathrm{BB} \right]\|_F^2=U.
\end{split} \label{eq:Opt}
\end{align}

The problem in \eqref{eq:Opt} is a mixed integer programming problem. Its solution requires a search over the entire $\cF^U \times \cW^U $ space of all possible $\bF_\mathrm{RF}$ and $\left\{\bw_u\right\}_{u=1}^U$ combinations. Further, the digital precoder $\bF_\mathrm{BB}$ needs to be jointly designed with the analog beamforming/combining vectors. In practice, this may require the feedback of the channel matrices $\bH_u, u=1, 2, ..., U$, or the effective channels, $\bw_u^* \bH_u \bF_\mathrm{RF}$. Therefore, the solution of \eqref{eq:Opt} requires large training and feedback overhead. Moreover, the optimal digital linear precoder is not known in general even without the RF constraints, and only iterative solutions exist \cite{Coordinated, Mazzarese}. Hence, the direct solution of this sum-rate maximization problem is neither practical nor tractable.

Similar problems to \eqref{eq:Opt} have been studied before in literature, but with baseband (not hybrid) precoding and combining \cite{Coordinated, Mazzarese,spencer2004zero,Boccardi_BD,Love_feedback,Jindal_Comb,Trivellato}. The main directions of designing the precoders/combiners in \cite{Coordinated, Mazzarese,spencer2004zero,Boccardi_BD,Jindal_Comb,Trivellato} can be summarized as follows.
\begin{itemize}
\item{\textbf{Iterative Coordinated Beamforming Designs}
     The general idea of these algorithms is to iterate between the design of the precoder and combiners in multi-user MIMO downlink systems, with the aim of converging to a good solution \cite{Coordinated, Mazzarese}. These algorithms, however, require either the availability of global channel knowledge at the transmitter, or the online BS-MS iterations to build the precoders and combiners. In mmWave systems, the application of coordinated beamforming is generally difficult as feeding the large mmWave channel matrix back to the BS requires a huge feedback overhead. Moreover, coordinated beamforming usually depends on using matching vectors at the MS's which can not be perfectly done with hybrid analog/digital architectures due to the hardware limitations on the analog precoders. Further, the convergence of coordinated beamforming has been established only for digital precoders \cite{Coordinated, Mazzarese}, and the extension to hybrid precoders has not yet been studied.}

\item{ \textbf{Non-iterative Designs with Channel State Information at the Transmitter} To avoid the design complexity associated with iterative methods, some non-iterative sub-optimal algorithms, like block diagonalization, were proposed \cite{spencer2004zero,Boccardi_BD}. Block diagonalization, however,requires global channel knowledge at the transmitter, which is difficult to achieve at mmWave systems. Further, the hardware constraints on the analog (or hybrid) precoding make it difficult to  exactly design the pre-processing matrix to have no  multi-user interference.}

\item{\textbf{Non-iterative Designs with Channel State Information at the Receiver} The main idea of these schemes is to first combine the MIMO channel at each receiver according to a certain criterion. Then, each each user quantizes its effective channel based on a pre-defined codebook, and feeds it back to the BS which uses it to construct its multi-user precoding matrix \cite{Jindal_Comb,Trivellato}. The application of these precoding/combining algorithms in mmWave systems is generally difficult because of the large dimensions of the mmWave channel matrix which makes the assumption of its availability at the MS's difficult to achieve in practice. Further, the hardware constraints make the direct application of the combining vector design schemes in \cite{Trivellato,Jindal_Comb} generally infeasible. }
\end{itemize}

Given the practical difficulties associated with applying the prior precoding/combining algorithms in mmWave systems, we propose a new mmWave-suitable multi-user MIMO beamforming algorithm in \sref{sec:MU_Alg}. Our proposed algorithm is developed to achieve good performance compared with the solution of \eqref{eq:Opt}, while requiring (i) low training overhead and (ii) small feedback overhead. After explaining the developed algorithm in \sref{sec:MU_Alg}, its performance is analyzed in \sref{sec:Performance} assuming infinite-resolution feedback and neglecting channel estimation errors. The performance degradations due to limited feedback are then analyzed in \sref{sec:Finite}.

\section{Two-stage Multi-user Hybrid Precoding Algorithm} \label{sec:MU_Alg}
The additional challenge in solving (7), beyond the usual coupling between precoders and combiners \cite{Coordinated, Mazzarese,spencer2004zero,Boccardi_BD,Jindal_Comb,Trivellato}, is the splitting of the precoding operation into two different domains, each with different constraints. The main idea of the proposed algorithm is to divide the calculation of the precoders into two stages. In the first stage, the BS RF precoder and the MS RF combiners are jointly designed to maximize the desired signal power of each user, neglecting the resulting interference among users. In the second stage, the BS digital precoder is designed to manage the multi-user interference.

\begin{algorithm} [!t]                     
\caption{Two-Stage Multi-user Hybrid Precoders Algorithm}          
\label{alg:MU_Precoding}                           
\begin{algorithmic} 
    \State \textbf{Input:} $\cF$ BS RF beamforming codebook of size $\left|\cF\right|=2^{B^\mathrm{BS}_\mathrm{RF}}$
    \State \hspace{28pt} $\cW$ MS RF beamforming codebook of size $\left|\cW\right|=2^{B^\mathrm{MS}_\mathrm{RF}}$
    \State \textbf{First stage:} Single-user RF beamforming/combining design
    \State  For each MS $u, u=1, 2, ..., U$
    \State \hspace{20pt} The BS and MS $u$ select  $\bv_u^\star$ and $\bg_u^\star$ that solve
    \State \hspace{80pt} $\left\{\bg_u^\star, \bv_u^\star\right\}=\displaystyle{\operatorname*{\arg max}_{ \substack{\forall \bg_u \in \cW \\ \forall \bv_u \in \cF }}}{\|\bg_u^* \bH_u \bv_u\|} $
    \State \hspace{20pt} MS $u$ sets $\bw_u=\bg_u^\star$
    \State BS sets $\bF_\mathrm{RF}=[\bv_1^\star, \bv_2^\star, ..., \bv_U^\star]$
    \State \textbf{Second stage:} Multi-user digital precoding design
    \State For each MS $u, u=1, 2, ..., U$
    \State \hspace{20pt} MS $u$ estimates its effective channel $\overline{\bh}_u=\bw_u^* \bH_u \bF_\mathrm{RF}$
    \State \hspace{20pt} MS $u$ quantizes $\overline{\bh}_u$ using a codebook $\cH$ of size $2^{B_\mathrm{BB}}$ and feeds back $\widehat{\bh}_u$ where
    \State $\hspace{80 pt} \widehat{\bh}_u=\displaystyle{\operatorname*{\arg max}_{\hat{\bh}_u \in \cH}}{\|\overline{\bh}_u^* \widehat{\bh}_u\|} $
    \State BS designs $\bF_\mathrm{BB}= {\widehat{\bH}}^*\left(\widehat{\bH} {\widehat{\bH}}^*\right)^{-1}$ with $\widehat{\bH}=\left[\widehat{\bh}_1^\mathrm{T},..., \widehat{\bh}_U^\mathrm{T}\right]^\mathrm{T}$
    \State $\bff_u^\mathrm{BB}= \frac{\bff_u^\mathrm{BB}}{\left\|\bF_\mathrm{RF} \bff_u^\mathrm{BB}\right\|_F}, u=1, 2, ..., U$
    \end{algorithmic}
\end{algorithm}

Algorithm \ref{alg:MU_Precoding} can be summarized as follows. In the first stage, the BS and each MS $u$ design the RF beamforming and combining vectors, $\bff^\mathrm{RF}_u$ and $\bw_u$, to maximize the desired signal power for user $u$, and neglecting the other users' interference. As this is the typical single-user RF beamforming design problem, efficient beam training algorithms developed for single-user systems such as \cite{Wang1,Multilevel}, which do not require explicit channel estimation and have a low training overhead, can be used to design the RF beamforming/combining vectors.

In the second stage, the BS trains the effective channels, $\overline{\bh}_u=\bw_u^* \bH_u \bF_\mathrm{RF}, u=1, 2, ..., U$, with the MS's. Note that the dimension of each effective channel vector is $U \times 1$ which is much less than the original channel matrix. This is not the case for the algorithms developed in \cite{Jindal_Comb,Trivellato} in which the effective channels have larger $N_\mathrm{BS} \times 1$ dimensions. Then, each MS $u$ quantizes its effective channel using a codebook $\cH$, and feeds the index of the quantized channel vector back to the BS with $B_\mathrm{BB}$ bits. Finally, the BS designs its zero-forcing digital precoder based on the quantized channels. Thanks to the narrow beamforming and the sparse mmWave channels, the effective MIMO channel is expected to be well-conditioned \cite{MassiveMIMO_Conv, Conjugate_BF}, which makes adopting a simple multi-user digital beamforming strategy like zero-forcing capable of achieving near-optimal performance \cite{ZF_Optmality}, as will be shown in Sections \ref{sec:Performance}-\ref{sec:Finite}.

Both the separate and joint designs of the analog and digital precoders were investigated before for single-user mmWave systems. For example, the work in \cite{Multi_beam} considered a single-user single-stream MIMO-OFDM system, where the analog and digital precoders were sequentially designed to maximize either the received signal strength or the sum-rate over different frequency sub-carriers. Alternatively, the analog and digital precoders were jointly designed in \cite{ayach2013spatially,Multi_beam,alkhateeb} to maximize the rate of single-user systems. In this paper, we consider a different setup which is multi-user downlink transmission. Therefore, the objective of the hybrid analog/digital beamforming in our work is different than that in \cite{ayach2013spatially,Multi_beam,alkhateeb} as we need to manage the multi-user interference as well. This leads to a completely different analysis.

In the next two sections, we analyze the performance of the proposed multi-user hybrid  precoding algorithm in different settings. For this analysis, we adopt the beamsteering codebook for the design of the analog beamforming/combining vectors. We also assume that the effective channels in the second stage of Algorithm \ref{alg:MU_Precoding} are quantized using a random vector quantization (RVQ) codebook. RVQ simplifies the analytical performance analysis of the proposed algorithm and allows leveraging some results from the limited feedback MIMO literature \cite{Jindal,Mazzarese,Love_feedback}.
\section{Performance Analysis with Infinite-Resolution Codebooks} \label{sec:Performance}
The analysis of hybrid precoding is non-trivial due to the coupling between analog and digital precoders. Therefore, we will study the performance of the proposed algorithm in two cases: With single-path channels and with large numbers of antennas. These  cases are of special interest as mmWave channels are likely to be sparse, i.e., only a few paths exist \cite{Rapp5G}, and both the BS and MS need to employ large antenna arrays to have sufficient received power \cite{ayach2013spatially}. Further, the analysis of these special cases will give useful insights into the performance of the proposed algorithms in more general settings which will also be confirmed by the simulations in \sref{sec:Results}.

In this section, we analyze the achievable rates of the proposed algorithm assuming perfect effective channel knowledge and supposing that the angles of the RF beamsteering vectors can take continuous values, i.e., we assume that both the RF codebooks ($\cF$ and $\cW$) and the RVQ  codebook $\cH$ are of infinite size. In \sref{sec:Finite}, we will study how limited feedback and finite codebooks affect the rates achieved by the developed hybrid precoding algorithm.
\subsection{Single-Path Channels} \label{subsec:single_path}

In this section, we consider the case when  $L_u=1, u=1, 2, ..., U$. For ease of exposition, we will omit the subscript $\ell$ in the definition of the channel parameters in \eqref{eq:channel_model}. The following theorem characterizes a lower bound on the achievable rate by each MS when Algorithm \ref{alg:MU_Precoding} is used to design the hybrid precoders at the BS and RF combiners at the MS's.

\begin{theorem}
Let Algorithm \ref{alg:MU_Precoding} be used to design the hybrid precoders and RF combiners described in \sref{sec:Model} under the following assumptions
\begin{enumerate}
\item{All channels are single-path, i.e., $L_u=1, u=1, 2, ..., U$.}
\item{The RF precoding vectors $\bff_u^\mathrm{RF}, u=1,2, ..., U$, and the RF combining vectors $\bw_u, u=1, 2, ..., U$ are beamsteering vectors with continuous angles.}
\item{Each MS $u$ perfectly knows its channel $\bH_u$, $u=1, 2, ..., U$. }
\item{The BS perfectly knows the effective channels $\overline{\bh}_u, u=1, 2, ..., U$.}
\end{enumerate}
and define the $N_\mathrm{BS} \times U$  matrix $\bA_\mathrm{BS}$ to gather the BS array response vectors associated with the $U$ AoDs, i.e.,  $\bA_\mathrm{BS}=\left[\ba_\mathrm{BS}\left(\phi_1\right), \ba_\mathrm{BS}\left(\phi_2\right), ..., \ba_\mathrm{BS}\left(\phi_U\right) \right]$, with maximum and minimum singular values $\sigma_{\mathrm{max}}(\bA_\mathrm{BS})$ and $\sigma_{\mathrm{min}}(\bA_\mathrm{BS})$, respectively.
Then, the achievable rate of user $u$ is lower bounded by
\begin{equation}
R_u \geq \log_2\left(1+ \frac{\mathsf{SNR}}{U}{N_\mathrm{BS} N_\mathrm{MS} \left|\alpha_u\right|^2} G\left(\left\{\phi_u\right\}_{u=1}^U\right) \right), \label{eq:lower_perfect}
\end{equation}%
where $G\left(\left\{\phi_u\right\}_{u=1}^U\right)= 4 \left(\frac{\sigma_{\mathrm{max}}^2\left(\bA_\mathrm{BS}\right)}{\sigma_{\mathrm{min}}^2\left(\bA_\mathrm{BS}\right)}+\frac{\sigma_{\mathrm{min}}^2\left(\bA_\mathrm{BS}\right)}{\sigma_{\mathrm{max}}^2\left(\bA_\mathrm{BS}\right)}+2\right)^{-1}$, $\mathsf{SNR}=\frac{P}{\sigma^2}$. \label{Th1}
\end{theorem}
\begin{proof}
Consider the BS and MS's with the system and channels described in \sref{sec:Model}. Then, in the first stage of Algorithm \ref{alg:MU_Precoding}, the BS and each MS $u$ find $\bv_u^\star$ and $\bg_u^\star$ that solve
\begin{equation}
\left\{\bg_u^\star, \bv_u^\star\right\}=\displaystyle{\operatorname*{\arg max}_{ \substack{\forall \bg_u \in \cW \\ \forall \bv_u \in \cF }}}{\|\bg_u^* \bH_u \bv_u\|}.
\label{eq:Stage1}
\end{equation}

As the channel $\bH_u$ has only one path, and given the continuous beamsteering capability assumption, the optimal RF precoding and combining vectors will be $\bg_u^\star=\ba_\mathrm{MS}(\theta_u)$, and $\bv_u^\star=\ba_\mathrm{BS}(\phi_u)$. Consequently, the MS sets $\bw_u=\ba_\mathrm{MS}(\theta_u)$ and the BS takes $\bff_u^\mathrm{RF}=\ba_\mathrm{BS}(\phi_u)$. Gathering the beamforming vectors for the $U$ users, the BS RF beamforming matrix is then $\bF_\mathrm{RF}=\bA_\mathrm{BS}=\left[\ba_\mathrm{BS}\left(\phi_1\right), \ba_\mathrm{BS}\left(\phi_2\right), ..., \ba_\mathrm{BS}\left(\phi_U\right) \right]$.

The effective channel for user $u$ after designing the RF precoders and combiners is
\begin{equation}
\begin{split}
\overline{\bh}_u&=\bw_u \bH_u \bF_\mathrm{RF}\\
&=\sqrt{N_\mathrm{BS} N_\mathrm{MS}} \alpha_u \ba^*_\mathrm{BS}\left(\phi_u\right) \bF_\mathrm{RF}.
\end{split}
\end{equation}

Now, defining $\overline{\bH}=[\overline{\bh}_1^\mathrm{T}, \overline{\bh}_2^\mathrm{T}, ..., \overline{\bh}_U^\mathrm{T}]^\mathrm{T}$, and given the design of $\bF_\mathrm{RF}$, we can write the effective channel matrix $\overline{\bH}$ as
\begin{equation}
\overline{\bH}=\bD \bA_\mathrm{BS}^* \bA_\mathrm{BS}, \label{eq:eff_channel}
\end{equation}
where $\bD$ is a $U \times U$ diagonal matrix with $\left[\bD\right]_{u,u}=\sqrt{N_\mathrm{BS} N_\mathrm{MS}} \alpha_u$.

Based on this effective channel, the BS zero-forcing digital precoder is defined as
\begin{equation}
\bF_\mathrm{BB}=\overline{\bH}^* \left(\overline{\bH} \overline{\bH}^*\right)^{-1} \boldsymbol\Lambda,
\end{equation}
where $\boldsymbol\Lambda$ is a diagonal matrix with the diagonal elements adjusted to satisfy the precoding power constraints $\left\|\bF_\mathrm{RF} \bff_u^\mathrm{BB}\right\|^2=1, u=1, 2, ..., U$. The diagonal elements of $\boldsymbol\Lambda$ are then equal to
\begin{equation}
{\boldsymbol\Lambda}_{u,u}=\sqrt{\frac{N_\mathrm{BS} N_\mathrm{MS} }{ \left(\bA_\mathrm{BS}^* \bA_\mathrm{BS}\right)^{-1}_{u,u}}}\left|\alpha_u\right|, u=1, 2, ..., U.
\end{equation}

Note that this $\boldsymbol\Lambda$ is different than the traditional digital zero-forcing precoder due to the different power constraints in the hybrid analog/digital architecture. [See Appendix \ref{app:power} for a derivation]

The achievable rate for user $u$ is then
\begin{align}
\begin{split}
R_u&=\log_2 \left(1+\frac{\mathsf{SNR}}{U} \left|\overline{\bh}_u^* \bff_u^\mathrm{BB}\right|^2\right),\\
&=\log_2\left(1+\frac{\mathsf{SNR}}{ U} \frac{N_\mathrm{BS} N_\mathrm{MS} \left|\alpha_u\right|^2 }{\left(\bA_\mathrm{BS}^* \bA_\mathrm{BS}\right)^{-1}_{u,u}}\right).
\end{split}
\end{align}

To bound this rate, the following lemma which characterizes a useful property of the matrix $\bA_\mathrm{BS}^* \bA_\mathrm{BS}$ can be used.
\begin{lemma}
Assume $\bA_\mathrm{BS}=\left[\ba_\mathrm{BS}\left(\phi_1\right), \ba_\mathrm{BS}\left(\phi_2\right), ..., \ba_\mathrm{BS}\left(\phi_U\right) \right]$, with the angles $\phi_u, u=1, 2, ..., U$ taking continuous values in $[0, 2 \pi]$, then the matrix $\bP=\bA_\mathrm{BS}^* \bA_\mathrm{BS}$ is positive definite almost surely.
\end{lemma}
\begin{proof}
Let the matrix $\bP = \bA_\mathrm{BS}^* \bA_\mathrm{BS}$, then for any non-zero complex vector $\bz \in \cC^{U}$, it follows that $\bz^* \bP \bz = \|\bA_\mathrm{BS} \bz\|_2^2 \geq 0$. Hence, the matrix $\bP$ is positive semi-definite. Further, if the vectors $\ba_\mathrm{BS}\left(\phi_1\right), \ba_\mathrm{BS}\left(\phi_2\right)$, $ ..., \ba_\mathrm{BS}\left(\phi_U\right)$ are linearly independent, then for any non-zero complex vector $\bz$, $\bA_\mathrm{BS} \bz \neq 0$, and the matrix $\bP$ is positive definite. To show that, consider any two vectors $\ba_\mathrm{BS}\left(\phi_u\right), \ba_\mathrm{BS}\left(\phi_n\right)$. These vectors are linearly dependent if and only if $\phi_u=\phi_n$. As the probability of this event equals zero when the AoDs $\phi_u$ and $\phi_n$ are selected independently from a continuous distribution, then the matrix $\bP$ is positive definite with probability one.
\end{proof}

Now, using the Kantorovich inequality \cite{kantorovich1948functional}, we can bound the diagonal entries of the matrix $\left(\bA_\mathrm{BS}^* \bA_\mathrm{BS}\right)^{-1}$ using the following lemma from \cite{bai1996bounds}.
\begin{lemma}
For any $n \times n$ Hermitian and positive definite matrix $\bP$ with the ordered eigenvalues satisfying $0<\lambda_\mathrm{min} \leq \lambda_2 \leq ... \leq \lambda_\mathrm{max}$, the element $\left(\bP \right)^{-1}_{u,u}, u=1, 2, ..., n$ satisfies
\begin{equation}
\left(\bP \right)^{-1}_{u,u} \leq \frac{1}{4 [\bP]_{u,u}}\left(\frac{\lambda_\mathrm{max}\left(\bP\right)}{\lambda_\mathrm{min}\left(\bP\right)}+\frac{\lambda_\mathrm{min}\left(\bP\right)}{\lambda_\mathrm{max}\left(\bP\right)}+2\right).
\end{equation}
\label{lemma3}
\end{lemma}

Finally, noting that $\left(\bA_\mathrm{BS}^* \bA_\mathrm{BS}\right)_{u,u}=1$, $\lambda_\mathrm{min}\left(\bA_\mathrm{BS}^* \bA_\mathrm{BS}\right)=\sigma_\mathrm{min}^2\left(\bA_\mathrm{BS}\right)$, and $\lambda_\mathrm{max}\left(\bA_\mathrm{BS}^* \bA_\mathrm{BS}\right)=\sigma_\mathrm{max}^2\left(\bA_\mathrm{BS}\right)$ and using lemma \ref{lemma3}, we get the lower bound on the achievable rate in \eqref{eq:lower_perfect}.
\end{proof}

In addition to characterizing a lower bound on the rates achieved by the proposed hybrid analog/digital precoding algorithm, the bound in \eqref{eq:lower_perfect}  separates the dependence on the channel gains $\alpha_u$, and the AoDs $\phi_u, u=1, 2, ..., U$ which can be used to claim the optimality of the proposed algorithm in some cases and to give useful insights into the  gain of the proposed algorithm over analog-only beamsteering solutions. This is illustrated in the following results.

\begin{proposition}
Denote the single-user rate as  $\mathring{R}_u=\log_2\left(1+\frac{\mathsf{SNR}}{ U}N_\mathrm{BS} N_\mathrm{MS} \left|\alpha_u\right|^2\right)$.  When Algorithm \ref{alg:MU_Precoding} is used to design the hybrid precoders and RF combiners described in \sref{sec:Model}, and given the assumptions stated in Theorem \ref{Th1}, the relation between the  achievable rate by any user $u$, and the single-user rate, $\mathring{R}_u$ satisfies
\begin{enumerate}
\item{$\bbE\left[ \mathring{R}_u- R_u\right]\leq K\left(N_\mathrm{BS},U\right)$}.
\item{$\lim_{N_\mathrm{BS}\rightarrow\infty}{R_u}=\mathring{R}_u$ almost surely}.
\end{enumerate}
where $K\left(N_\mathrm{BS},U\right)$ is a constant whose value depends only on $N_\mathrm{BS}$ and $U$.
\label{Prop1}
\end{proposition}
\begin{proof}
 See Appendix \ref{app:Cor1}.
\end{proof}

Proposition \ref{Prop1} indicates that the average achievable rate of any user $u$ using the proposed low-complexity precoding/combining algorithm grows with the same slope of the single-user rate at high SNR, and stays within a constant gap from it. This gap, $K\left(N_\mathrm{BS}, U\right)$, depends only on the number of users and the number of BS antennas. As the number of BS antennas increases, the matrix $\bA_\mathrm{BS}$ becomes more well-conditioned, and the ratio between its maximum and minimum singular values will approach one. Hence, the value of $G\left(\left\{\phi_u\right\}_{u=1}^U\right)$ in \eqref{eq:lower_perfect} will be closer to one, and the gap between the achievable rate using Algorithm \ref{alg:MU_Precoding} and the single-user rate will decrease. This will also be shown by numerical simulations in \sref{sec:Results}. One important note here is that this gap does not depend on the number of MS antennas, which is contrary to the analog-only beamsteering, given by the first stage only of Algorithm \ref{alg:MU_Precoding}. This leads to the following corollary.

\begin{corollary}
Let $R_\mathrm{BS}$ denote the rate achieved by user $u$ when the BS employs analog-only beamsteering designed according to Step $1$ of Algorithm \ref{alg:MU_Precoding}. Then, the relation between the average achievable rate using Algorithm \ref{alg:MU_Precoding} $R_u$ and the average rate of analog-only beamsteering solution when the number of MS antennas goes to infinity satisfies: $\lim_{N_\mathrm{MS}\rightarrow \infty}\bbE\left[R_u-R_\mathrm{BS}|\right]= \infty$.
\label{Prop2}
\end{corollary}
\begin{proof}See Appendix \ref{app:Prop2}.
\end{proof}

This corollary implies that multi-user interference management is still important at mmWave systems even when very large numbers of antennas are used at the BS and MS's. Note also that this is not the case when the number of BS antennas goes to infinity as it can be easily shown that the performance of RF beamsteering alone becomes optimal in this case.

\subsection{Large-dimensional Regime}\label{subsec:Perfect_large}
Under the assumption of large numbers of transmit antennas, a different approximation of the achievable rate can be derived. We approach this problem using the virtual channel model framework and its simplifications in large MIMO systems \cite{Virtual, Hong}. The results of this section are, therefore, valid only for uniform arrays, e.g., ULAs and UPAs \cite{Virtual,CAP}. The virtual channel model characterizes physical channels via joint spatial beams in \textit{fixed} virtual transmit and receive directions exploiting the finite dimensionality of the MIMO system, i.e., the finite number of transmit and receive antennas. The virtual transmit and receive directions are fixed because they depend only on the number of BS and MS antennas. Hence, they are common for the different users with the same number of antennas. Using this channel model, the $u$th user channel $\bH_u$ can be written as \cite{Virtual}
\begin{equation}
\bH_u=\overline{\bA}_\mathrm{MS} \bH_u^\mathrm{v} \overline{\bA}_\mathrm{BS}^*,
\label{eq:virtual}
\end{equation}
where $\overline{\bA}_\mathrm{BS}=[\ba_\mathrm{BS}\left(\bar{\phi}_{1}\right), \ba_\mathrm{BS}\left(\bar{\phi}_{2}\right), ..., \ba_\mathrm{BS}\left(\bar{\phi}_{N_\mathrm{BS}}\right)]$ is an $N_\mathrm{BS}\times N_\mathrm{BS}$ matrix carrying the BS array response vectors in the virtual directions $\bar{\phi}_{p}, p=1, 2, ..., N_\mathrm{BS}$ that satisfy $ \frac{2 \pi d}{\lambda}\sin\left(\bar{\phi}_{p}\right)=\frac{2 \pi p}{N_\mathrm{BS}}$. Similarly, $\overline{\bA}_\mathrm{MS}=[\ba_\mathrm{MS}\left(\bar{\theta}_{1}\right), \ba_\mathrm{MS}\left(\bar{\theta}_{2}\right), ..., \ba_\mathrm{MS}\left(\bar{\theta}_{N_\mathrm{MS}}\right)]$ carries the MS array response vectors in the virtual directions $\bar{\theta}_{q}, q=1, 2, ..., N_\mathrm{MS}$ that satisfy $\frac{2 \pi d}{\lambda}\sin\left(\bar{\theta}_{q}\right)=\frac{2 \pi q}{N_\mathrm{MS}}$. Thanks to these special virtual channel angles, the matrices $\overline{\bA}_\mathrm{BS}$ and $\overline{\bA}_\mathrm{MS}$ are DFT matrices \cite{Virtual}. Finally, $\bH_u^\mathrm{v}$ is the $u$th MS virtual channel matrix with each element $\left[\bH_u^\mathrm{v}\right]_{q,p}$ representing a group of physical spatial paths, and approximately equal to the sum of the gains of those paths \cite{Virtual}.

One advantage of using the virtual channel model in analyzing our proposed multi-user precoding algorithm lies in the fact that it provides a common space of the transmit eigenvectors of the different users. This means that the BS eigenvectors for each MS form a subset of the columns of the DFT matrix $\overline{\bA}_\mathrm{BS}$. The virtual channel model also provides a simple way to incorporate the angle spread associated with mmWave channel scatterers by defining each element of the virtual channel matrix as the sum of the channel gains associated with the scatterers located in a certain direction multiplied by the integration of the spatial spreading functions of these scatterers \cite{Virtual}.

Before leveraging this channel model in analyzing the proposed hybrid precoding algorithm, we rewrite the channel in \eqref{eq:virtual} as
\begin{equation}
\bH_u=\sqrt{\frac{N_\mathrm{BS} N_\mathrm{MS}} { L_{u}}} \sum_{m=1}^{N_\mathrm{BS} N_\mathrm{MS}}{\gamma_{u,m} \ba_\mathrm{MS}\left(\bar{\theta}_{u,m}\right) \ba_\mathrm{BS}^*\left(\bar{\phi}_{u,m}\right)},
\label{eq:virtual2}
\end{equation}
where $\left|\gamma_{u,1}\right| \geq \left|\gamma_{u,2}\right| \geq ... \geq \left|\gamma_{u,N_\mathrm{BS} N_\mathrm{MS}}\right|$. $\gamma_{u,m}$ equals the element in $\bH_u^\mathrm{v}$ with the $m$th largest virtual channel element magnitude, and $\bar{\phi}_{u,m}, \bar{\theta}_{u,m}$ are the corresponding transmit and receive virtual directions, respectively.

In the following proposition, we use this channel model to characterize a simple lower bound on the achievable rate of Algorithm \ref{alg:MU_Precoding} for arbitrary numbers of channel paths assuming for simplicity that $L_u=L, u=1, 2, ..., U$. The derived results give useful analytical insights into the asymptotic performance of the proposed algorithm in the multi-path case.

\begin{proposition}
\label{prop:Prop6}
Define the single-user rate as $\mathring{R}_u=\log2(1+ \frac{\mathsf{SNR}}{ U L} N_\mathrm{BS} N_\mathrm{NS} \left|\gamma_{u,1}\right|^2)$. Then, when Algorithm \ref{alg:MU_Precoding} is used to design the hybrid analog/digital precoders at the BS and RF combiners at the MSs, with the assumptions in Theorem \ref{Th1}, and adopting the virtual channel model in \eqref{eq:virtual2}, the average achievable rate of user $u$ is lower bounded by
\begin{equation}
\bbE\left[R_u\right]\geq \bbE\left[ \mathring{R}_u \right] \left(\prod_{i=1}^{U-1}\left(1-\frac{i}{N_\mathrm{BS}}\right) \left(1-\frac{L-1}{N_\mathrm{MS}}\right)^{U}+ \mathbbm{1}_{\left(L>1\right)} \prod_{i=1}^{U-1}\left(1-\frac{i L}{N_\mathrm{BS}}\right) \left(\frac{1}{N_\mathrm{MS}}\right)^{\left(L-1\right)U}\right).
\label{eq:vir_bound}\end{equation}
\end{proposition}
\begin{proof}
Consider the BS and MS's with the system model described in \sref{sec:Model}, and the approximated channel model in \eqref{eq:virtual2}. In the first stage of Algorithm \ref{alg:MU_Precoding}, the BS and each MS $u$ find $\bv_u^\star$  and $\bg_u^\star$  that solve \eqref{eq:Stage1}. Given the virtual channel model in \eqref{eq:virtual2}, we get $\bw_u=\bg_u^\star=\ba_\mathrm{MS}\left(\bar{\theta}_{u,1}\right)$ and $\bv_u^\star=\ba_\mathrm{BS}\left(\bar{\phi}_{u,1}\right)$. Consequently, the RF precoder at the BS becomes $\bF_\mathrm{RF}=\left[\ba_\mathrm{BS}\left(\bar{\phi}_{1,1}\right), \ba_\mathrm{BS}\left(\bar{\phi}_{2,1}\right), ..., \ba_\mathrm{BS}\left(\bar{\phi}_{U,1}\right)\right]$. Now, we can write the $u$th MS effective channel as
\begin{equation}
\overline{\bh}_u=\bw_u^* \bH_u \bF_\mathrm{RF}=\sqrt{\frac{N_\mathrm{BS} N_\mathrm{MS}}{ L}}\gamma_{u,1} \left[ \zeta_{u,1}, \zeta_{u,2}, ..., \zeta_{u,U}\right],
\label{eq:ver_eff}
\end{equation}
where the values of the $\zeta_{u,n}$ elements are
\begin{itemize}
\item {$\zeta_{u,u}=1$},
\item {$\displaystyle{\zeta_{u,n}=\mathbbm{1}_{\left(\bar{\phi_{u,1}}=\bar{\phi}_{n,1}\right)}+\sum_{m=2}^L \frac{\gamma_{u,2}}{\gamma_{u,1}} \mathbbm{1}_{\left(\bar{\phi}_{u,m}=\bar{\phi}_{n,1}\right)} \mathbbm{1}_{\left(\bar{\theta}_{u,1}=\bar{\theta}_{u,m}\right)}}, \hspace{10pt} \forall n \neq u$},
\end{itemize}
where the summation in $\zeta_{u,n}$ is over the first $L$ elements only due to the sparse channel. Note that the characterization of $\zeta_{u,n}$ is due to the DFT structure of the matrices $\overline{\bA}_\mathrm{BS}$ and $\overline{\bA}_\mathrm{MS}$.

The overall effective channel, $\overline{\bH}$, can be then written as
\begin{equation}
\overline{\bH}=\bD_\mathrm{v} \bP_\mathrm{v},
\end{equation}
where $\bD_\mathrm{v}$ is a diagonal matrix with the diagonal elements $\left[\bD_\mathrm{v}\right]_{u,u}=\sqrt{\frac{N_\mathrm{BS} N_\mathrm{MS}}{ L}}\gamma_{u,1}, u=1, 2, ..., U$, and the $U \times U$ matrix $\bP_\mathrm{v}$ has $\left[\bP_\mathrm{v}\right]_{u,n}=\zeta_{u,n}, \forall u,n$.

The digital zero-forcing precoder is therefore $\bF_\mathrm{BB}=\overline{\bH}^* \left( \overline{\bH}  \overline{\bH}^*  \right)^{-1} \boldsymbol\Lambda$, and the diagonal elements of $\boldsymbol\Lambda$ are chosen to satisfy the precoding power constraint $\left\|\bF_\mathrm{RF} \bff_u^\mathrm{BB}\right\|^2=1$. Using a similar derivation to that in Appendix \ref{app:power}, we get
\small
\begin{equation}
\left[\boldsymbol\Lambda\right]_{u,u}=\frac{\sqrt{\frac{N_\mathrm{BS}N_\mathrm{MS}}{ L}} \left|\gamma_{u,1}\right|}{\sqrt{\left( \left(\bP_\mathrm{v} \bP_\mathrm{v}^*\right)^{-1} \bP_\mathrm{v} \bF_\mathrm{RF}^* \bF_\mathrm{RF} \bP_\mathrm{v}^*  \left(\bP_\mathrm{v} \bP_\mathrm{v}^*\right)^{-1}\right)_{u,u}}}.
\label{eq:ppp}
\end{equation}

\normalsize
Using the designed digital and analog precoders, the rate of user $u$ can be written as
\small
\begin{align}
\bbE\left[R_u\right]&=\bbE\left[\log_2\left(1+\frac{\mathsf{SNR}}{U}\left|\overline{\bh}_u^* \bff_u^\mathrm{BB} \right|^2\right)\right],\\
&= \bbE\left[\log_2\left(1+\frac{\mathsf{SNR}}{U L } \frac{N_\mathrm{BS} N_\mathrm{MS} \left|\gamma_{u,1}\right|^2}{ \left( \left(\bP_\mathrm{v} \bP_\mathrm{v}^*\right)^{-1} \bP_\mathrm{v} \bF_\mathrm{RF}^* \bF_\mathrm{RF} \bP_\mathrm{v}^*  \left(\bP_\mathrm{v} \bP_\mathrm{v}^*\right)^{-1}\right)_{u,u}}\right)\right]. \label{eq:LargePerfect}
\end{align}

\normalsize
Now, we note that the term $\left( \left(\bP_\mathrm{v} \bP_\mathrm{v}^*\right)^{-1} \bP_\mathrm{v} \bF_\mathrm{RF}^* \bF_\mathrm{RF} \bP_\mathrm{v}^*  \left(\bP_\mathrm{v} \bP_\mathrm{v}^*\right)^{-1}\right)_{u,u}=1$ if $\bP_\mathrm{v}=\bI$. Then, considering only the case when $\bP_\mathrm{v}=\bI$  gives a simple lower bound on the achievable rate
\small
\begin{align}
\bbE\left[R_u\right] & \geq \bbE\left[\log_2\left(1+\frac{\mathsf{SNR}}{ U L } N_\mathrm{BS} N_\mathrm{MS} \left|\gamma_{u,1}\right|^2\right) \mathbbm{1} \left(\bP_\mathrm{v}=\bI_U \right)\right], \\
&\stackrel{(a)}{=} \bbE\left[\log_2\left(1+\frac{\mathsf{SNR}}{ U L } N_\mathrm{BS} N_\mathrm{MS} \left|\gamma_{u,1}\right|^2\right)\right] \mathrm{P} \left(\bP_\mathrm{v}=\bI_U \right), \label{eq:Ach}
\end{align}
\normalsize
where (a) is by leveraging the independence between $\gamma_{u,1}$ and the virtual transmit angles of the different users. Thanks to the sparse nature of mmWave channels, this simple bound in \eqref{eq:Ach} can be a tight bound on the achievable rate. Finally, the probability of the event $\bP_\mathrm{v}=\bI$ can be bounded as follows by considering only the cases when all the AoAs are equal or all of them are different
\small
\begin{equation}
\begin{split}
\mathrm{P} \left(\bP_\mathrm{v}=\bI_U \right) &\geq \bbP \left(\bP_\mathrm{v}=\bI_U \left| \bigcap_{u=1}^U \left(\bar{\theta}_{u,1} \neq \bar{\theta}_{u,m}, \forall m \neq 1\right) \right.\right) \mathrm{P} \left(\bigcap_{u=1}^U \left(\bar{\theta}_{u,1} \neq \bar{\theta}_{u,m}, \forall m \neq 1\right) \right) \\
& \hspace{30 pt} + \mathrm{P} \left(\bP_\mathrm{v}=\bI_U \left| \bigcap_{u=1}^U \left(\bar{\theta}_{u,1} = \bar{\theta}_{u,m}, \forall m \neq 1\right) \right) \mathrm{P} \left(\bigcap_{u=1}^U \left(\bar{\theta}_{u,1} = \bar{\theta}_{u,m}, \forall m \neq 1\right) \right.\right),
\end{split}
\end{equation}
\begin{equation}
 \hspace{20pt} \geq \prod_{i=1}^{U-1}\left(1-\frac{i}{N_\mathrm{BS}}\right) \left(1-\frac{L-1}{N_\mathrm{MS}}\right)^{U}+ \mathbbm{1}_{\left(L>1\right)} \prod_{i=1}^{U-1}\left(1-\frac{i L}{N_\mathrm{BS}}\right) \left(\frac{1}{N_\mathrm{MS}}\right)^{\left(L-1\right)U}, \label{eq:ProbI}
\end{equation}
\normalsize
where all these probabilities are calculated from the expression of $\zeta_{u,n}, n \neq u$ (the off-diagonal entries of $\bP_\mathrm{v}$).
\end{proof}

This bound shows the asymptotic optimality of the sum-rate achieved by the proposed hybrid precoding algorithm in the large-dimensional regime, as it approaches $1$ with large numbers of antennas. Hence, the average achievable rate by the proposed algorithm in \eqref{eq:Ach} will be very close to the single-user rate. Indeed, this simple bound can be shown to be tight when the number of paths is very small compared with the number of antennas which is the case in mmWave systems. Also, this bound shows the relatively small importance of the other paths, rather than the strongest path, on the performance as $\frac{L}{N_\mathrm{BS}}\ll 1$ and $\frac{L-1}{N_\mathrm{MS}}\ll 1$. Finally, note that the bound in \eqref{eq:vir_bound} is an approximated bound, as it depends on the asymptotic properties of the virtual channel model in \eqref{eq:virtual2}, which becomes a good approximation when the number of antennas is very large.
\section{Rate Loss with Limited Feedback} \label{sec:Finite}
In this section, we consider RF and digital codebooks with \textit{finite} sizes, and analyze the rate loss due to the joint RF/baseband quantization. Although the analysis will consider the special cases of single-path mmWave channels, and large-dimensional regimes, it helps making important conclusions about the performance of the hybrid precoding over finite-rate feedback channels.
\subsection{Single-Path Channels} \label{subsec:Limited_single_path}
Considering single-path mmWave channels, the following theorem characterizes the average rate loss when the hybrid analog/digital precoders and RF combiners are designed according to Algorithm \ref{alg:MU_Precoding} with the quantized beamsteering RF precoders $\cF, \cW$, and the effective channel RVQ codebook $\cH$.

\begin{theorem}
Let $R_u^\mathrm{Q}$ denote the rate achieved by user $u$ when Algorithm \ref{alg:MU_Precoding} is used to design the hybrid precoders and RF combiners described in \sref{sec:Model} while assuming that
\begin{enumerate}
\item{All channels are single-path, i.e., $L_u=1, u=1, 2, ..., U$.}
\item{The RF precoding and combining vectors, $\bff_u^\mathrm{RF}, u=1,2, ..., U$ and  $\bw_u, u=1, 2, ..., U$, are beamsteering vectors selected from the quantized codebooks $\cF$ and $\cW$.}
\item{Each MS $u$ perfectly knows its channel $\bH_u$, $u=1, 2, ..., U$. }
\item{Each MS $u$ quantizes its effective channel $\overline{\bh}_u$ using a RVQ codebook $\cH$ of size $\left|\cH\right|=2^{B_\mathrm{BB}}$.}
\end{enumerate}

Recall that $R_u$ is the rate achieved by user $u$ with the assumptions in Theorem \ref{Th1}. Then the average rate loss per user, $\overline{\Delta R}_u=\bbE\left[R_u-R_u^\mathrm{Q}\right]$, is upper bounded by
\begin{equation}
\overline{\Delta R}_u \leq \log_2\left(\frac{1+\frac{\mathsf{SNR}}{ U} N_\mathrm{BS} N_\mathrm{MS}\bar{\alpha} \left(1+\frac{U-1}{N_\mathrm{BS}}\right) 2^{-\frac{B_\mathrm{BB}}{U-1}}}{\left|\overline{\mu}_\mathrm{BS}\right|^2 \left|\overline{\mu}_\mathrm{MS}\right|^2}\right),
\end{equation}
where $\left|\overline{\mu}_\mathrm{BS}\right|=\displaystyle{\min_{\bff_u \in \cF} \max_{\bff_n \in \cF}} {\left| \bff_u^* \bff_n\right|}$, and $\left|\overline{\mu}_\mathrm{MS}\right|=\displaystyle{\min_{\bw_u \in \cW} \max_{\bw_n \in \cW}} {\left| \bw_u^* \bw_n\right|}$.
\label{Th2}
\end{theorem}
\begin{proof}
See appendix \ref{app:Prop_Finite}.
\end{proof}

Theorem \ref{Th2} characterizes an upper bound on the rate loss due to quantization. It can be used to determine how the number of baseband and RF quantization bits should scale with the different system and channel parameters to be within a constant gap of the optimal rate. This is captured in the following corollary.

\begin{corollary}
To maintain a rate loss of $\log_2\left(b\right)$ bps/Hz per user, the number of baseband quantization bits should satisfy
\small
\begin{equation}
B_\mathrm{BB}=\frac{U-1}{3} \mathsf{SNR}_\mathrm{dB} + (U-1) \log_2 \left(\frac{N_\mathrm{BS}N_\mathrm{MS}}{ U}\bar{\alpha} \left(1-\frac{U-1}{N_\mathrm{BS}}\right)\right)-(U-1) \log_2\left(\left|\overline{\mu}_\mathrm{BS}\right|^2 \left|\overline{\mu}_\mathrm{MS}\right|^2 b -1\right).
\end{equation}
\normalsize
\end{corollary}

This corollary shows that the number of bits used to quantize the effective channels should increase linearly with the SNR in dB for any given number of users and logarithmically with the number of antennas. It also illustrates that more baseband quantization bits will be needed if the RF beamsteering vectors are poorly quantized, i.e., if $\left|\overline{\mu}_\mathrm{BS}\right|$ and $\left|\overline{\mu}_\mathrm{MS}\right|$ are small.

The relation between the RF and baseband quantization bits is important to understand the behavior of  hybrid precoding algorithms. Indeed, in some cases, e.g., when the effective channel is poorly quantized, the performance of analog-only beamforming can exceed that of the hybrid precoding. In \sref{sec:Results}, the hybrid precoding and beamsteering algorithms are compared for different quantization settings, and some insights are given to highlight the cases in which using a digital layer to manage the multi-user interference is useful.
\subsection{Large-dimensional Regime}\label{subsec:Limited_Large}
When large antenna arrays are used at both the BS and MS's, using the virtual channel model in \sref{subsec:Perfect_large}, we can bound the average rate loss using the proposed hybrid precoding algorithm with finite size codebooks.

\begin{figure}[t]
\centering
\subfigure[center][{}]{
\includegraphics[width=0.470\columnwidth]{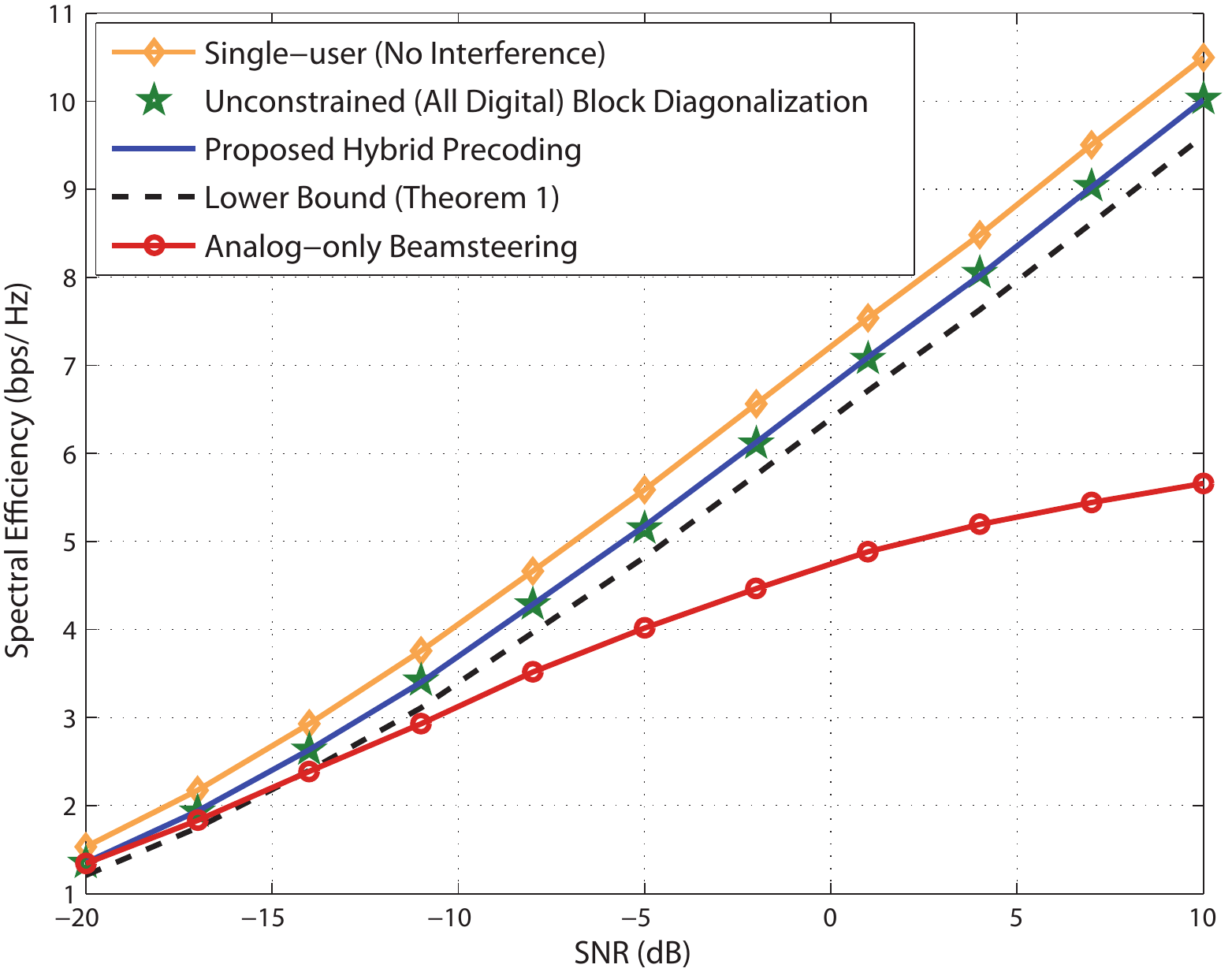}
\label{fig:Perfect_SNR}}
\subfigure[center][{}]{
\includegraphics[width=0.470\columnwidth]{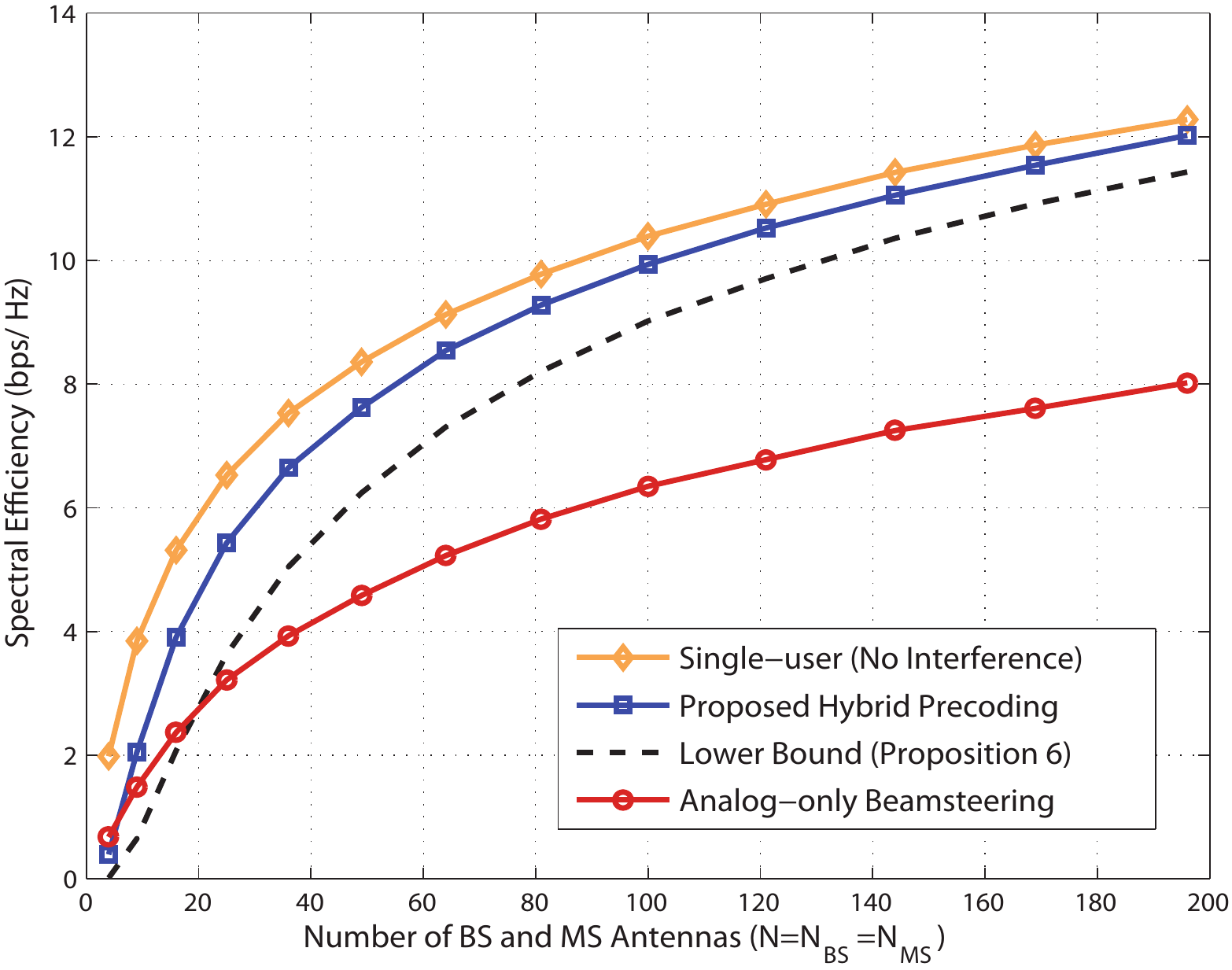}
\label{fig:Large_Arrays}}
\caption{Achievable rates using the hybrid  precoding and beamsteering algorithms with perfect channel knowledge. Single-path channels are assumed in (a), while channels with $L=3$ paths are examined in (b). }
\label{fig:Fig3}
\vspace{-10pt}
\end{figure}

\begin{proposition}
Using Algorithm \ref{alg:MU_Precoding} to design the hybrid precoders at the BS and RF combiners at the MSs, with the assumptions in Theorem \ref{Th2}, and adopting the virtual channel model in \eqref{eq:virtual2}, the average rate loss per user due to quantization, $\overline{\Delta R}_u=\bbE\left[R_u-R_u^\mathrm{Q}\right]$, is upper bounded by
\begin{equation}
\overline{\Delta R}_u \leq \log_2\left(1+\frac{\mathsf{SNR}}{ U} \bar{\alpha} N_\mathrm{BS} N_\mathrm{MS} \left(1+\frac{U-1}{N_\mathrm{BS}}\left(1+\frac{L-1}{N_\mathrm{BS}N_\mathrm{MS}}\right)\right) 2^{-\frac{B_\mathrm{BB}}{U-1}}\right).
\end{equation}
\end{proposition}

The proof is similar to Theorem \ref{Th2}, but leverages the definition of the effective channel in \eqref{eq:ver_eff}. In addition to characterizing the rate loss due to quantization for more general settings with multi-path mmWave channels, this result illustrates the marginal impact of the other paths on the performance of mmWave systems as $\frac{L-1}{N_\mathrm{BS}N_\mathrm{MS}}\ll 1$. In other words, this indicates that considering only the path with the maximum gain gives a very good  performance.


\section{Simulation Results} \label{sec:Results}


In this section, we evaluate the performance of the proposed hybrid analog/digital precoding algorithm and derived bounds using numerical simulations. All the plotted rates in \figref{fig:Fig3}-\figref{fig:Fig6} are the averaged achievable rates per user; $\bbE\left[\frac{1}{U} \sum_{u=1}^{U} R_u\right]$ with $R_u$ in equation \ref{eq:Rate}.

First, we compare the achievable rates without quantization loss and with perfect effective channel knowledge in \figref{fig:Perfect_SNR} and \figref{fig:Large_Arrays}. In \figref{fig:Perfect_SNR}, we consider the system model in \sref{sec:Model} with a BS employing an $8 \times 8$ UPA with $4$ MS's, each having a $4 \times 4$ UPA. The channels are single-path, the azimuth AoAs/AoDs are assumed to be uniformly distributed in $[0, 2 \pi ]$, and the elevation AoAs/AoDs are uniformly distributed in $[- \frac{\pi}{2}, \frac{\pi}{2}]$. The SNR in the plots is defined as $\text{SNR}=\frac{P \bar{\alpha}}{\sigma^2 U}$. The rate achieved by the proposed hybrid precoding/combining algorithm is compared with the single-user rate and the rate obtained by beamsteering. These rates are also compared with the performance of the unconstrained (all digital) block diagonalization with dominant eigenmode transmission \cite{spencer2004zero}. Block diagonalization requires that $N_\mathrm{BS}-\mathrm{rank}\left(\left[\bH_1^T \ ...\ \bH_{u-1}^T \bH_{u+1}^T \ ... \ \bH_U^T\right]^T\right)>0, \ \forall u$ which is expected to be satisfied with high probability in mmWave systems with large arrays and sparse channels. The implementation of block diagonalization using hybrid precoding, however, is difficult due to the reasons mentioned in \sref{sec:Form}. The figure indicates that the performance of hybrid precoding is very close to the single-user rate thanks to cancelling the residual multi-user interference. It also shows that the proposed hybrid precoding algorithm achieves almost the same performance of the unconstrained block diagonalization, and offers a good gain over analog-only beamsteering solutions. The gain over beamsteering increases with SNR as the beamsteering rate starts to be interference limited. The tightness of the derived lower bound in Theorem \ref{Th1} is also shown.

\begin{figure}[t]
\centering
\subfigure[center][{}]{
\includegraphics[width=0.47\columnwidth]{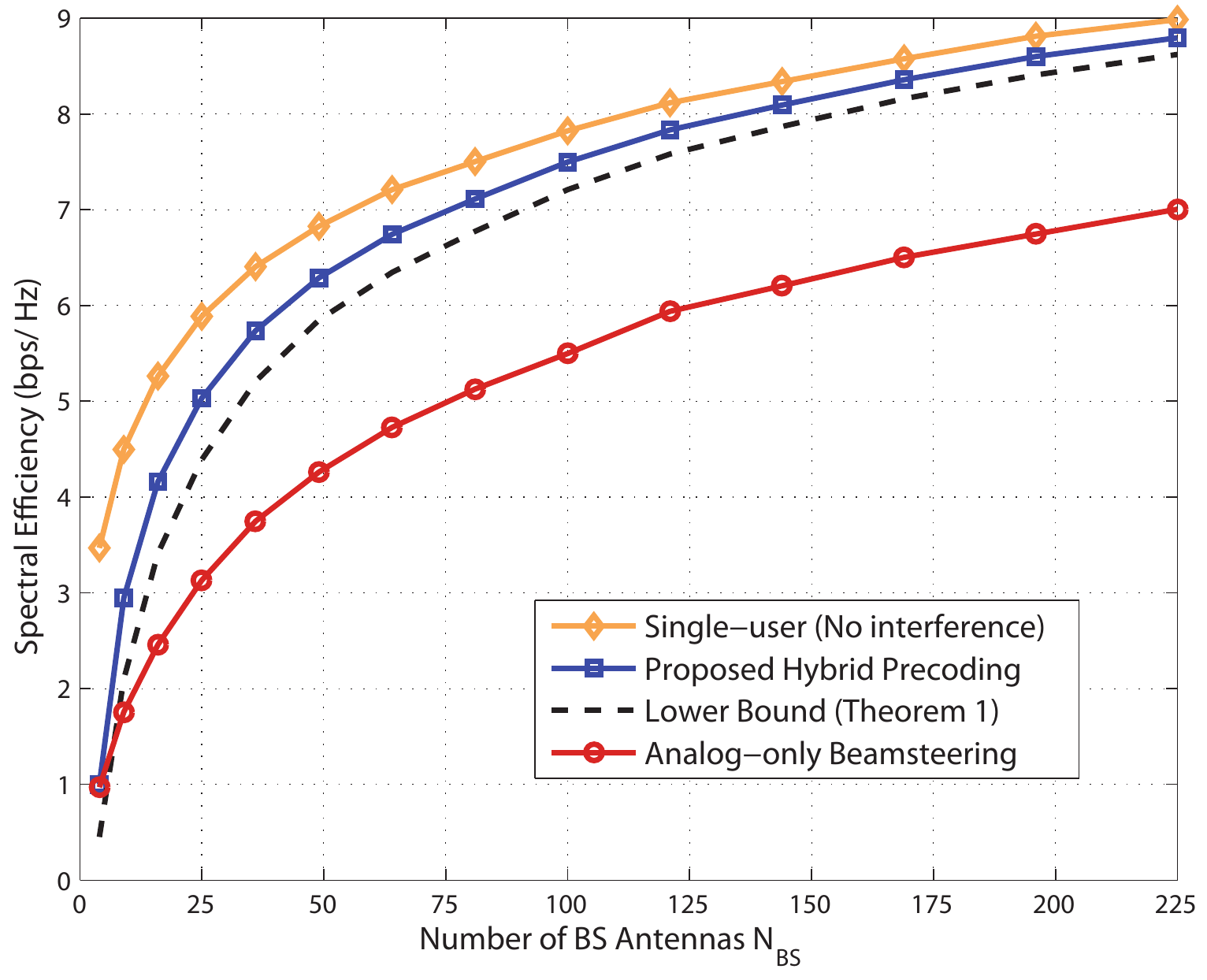}
\label{fig:Perfect_Nbs}}
\subfigure[center][{}]{
\includegraphics[width=0.473\columnwidth]{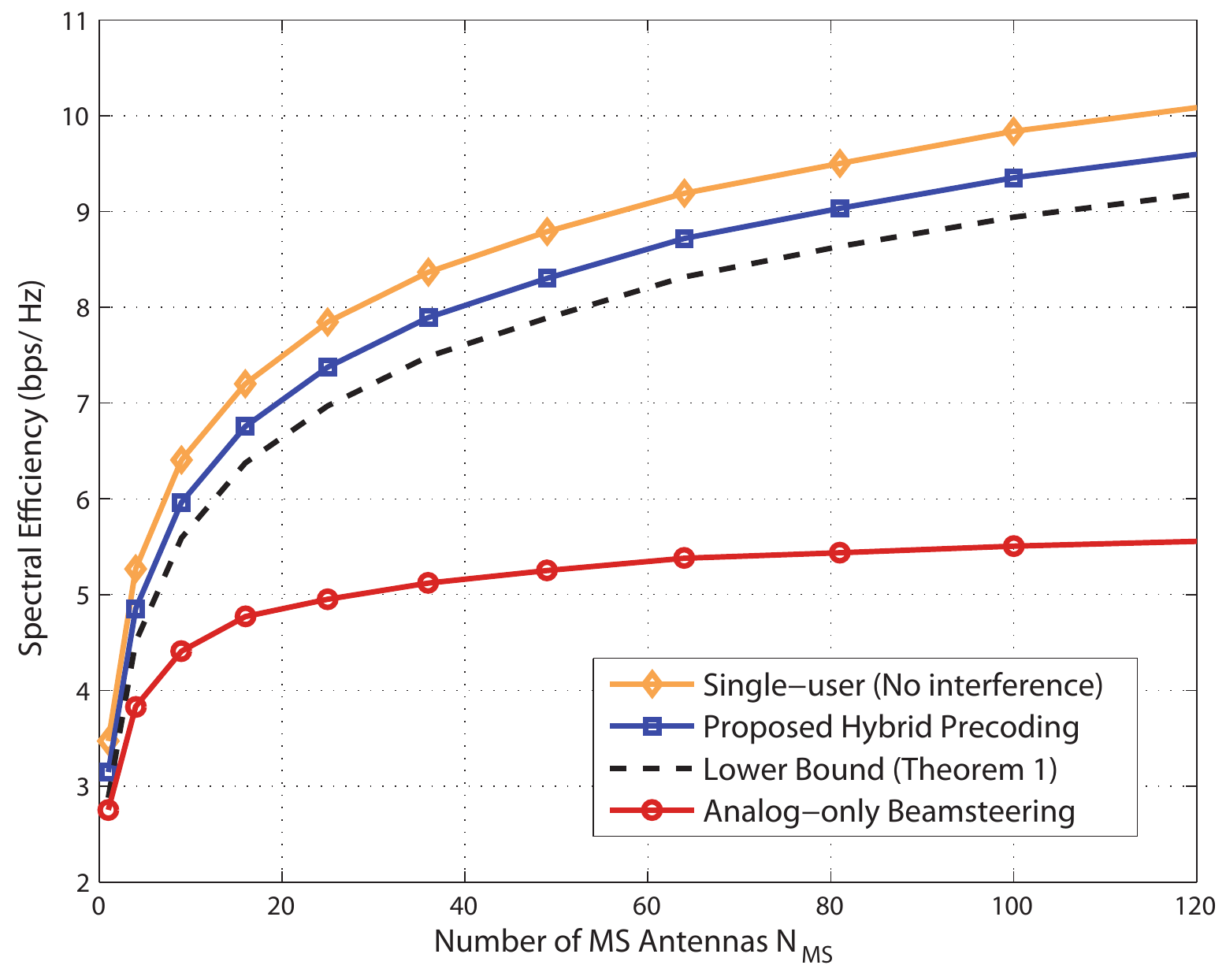}
\label{fig:Perfect_Nms}}
\caption{Achievable rates using the hybrid  precoding and beamsteering algorithms with perfect channel knowledge. In (a), the performance of hybrid precoding is shown to approach the single-user rate with large numbers of BS antennas. In (b), the performance gap between hybrid precoding and beamsteering increases with more MS antennas. }
\label{fig:Fig4}
\vspace{-10pt}
\end{figure}

In \figref{fig:Large_Arrays}, we consider the same setup, but when each channel has $L=3$ paths. The rates of the single-user, hybrid precoding, and beamsteering are simulated with different numbers of BS and MS antennas, assuming that $N_\mathrm{BS}=N_\mathrm{MS}$. The bound derived in Proposition \ref{prop:Prop6} was also plotted where it is shown to be tight at large number of antennas as discussed in \sref{subsec:Perfect_large}.

In \figref{fig:Perfect_Nbs}, the same setup in \figref{fig:Perfect_SNR} is considered  at SNR$=0$ dB, but with different values of BS antennas. The figure shows that even at very large numbers of antennas, there is still a considerable gain of hybrid precoding over beamsteering. This figure also shows that the difference between hybrid precoding and the single-user rate decreases at a large number of BS antennas which validates the second part of Proposition \ref{Prop1}.

In \figref{fig:Perfect_Nms}, the same setup is considered with an $8 \times 8$ BS UPA and with different numbers of MS antennas. The figure illustrates how the performance gap between hybrid precoding and beamsteering increases with increasing the number of MS antennas which coincides with Corollary \ref{Prop2}. This means that hybrid precoding has a higher gain over analog-only beamforming solutions in mmWave systems when large antenna arrays are employed at the MS's.

To illustrate the impact of RF quantization, the performance of hybrid precoding and analog-only beamsteering are evaluated in \figref{fig:AQ} with different numbers of quantization bits at the BS and MS. We consider the same setup of \figref{fig:Perfect_SNR} with $4 \times 4$ MS UPAs and when each channel has $L=3$ paths. As shown in the figure, the performance of the beamforming strategies degrades with decreasing the number of quantization bits. The gain, however, stays almost constant for the same number of antennas. The figure also shows that the number of quantization bits should increase with the antenna numbers to avoid significant performance degradations.
\begin{figure}[t]
\centering
\subfigure[center][{}]{
\includegraphics[width=0.47\columnwidth, height=180pt]{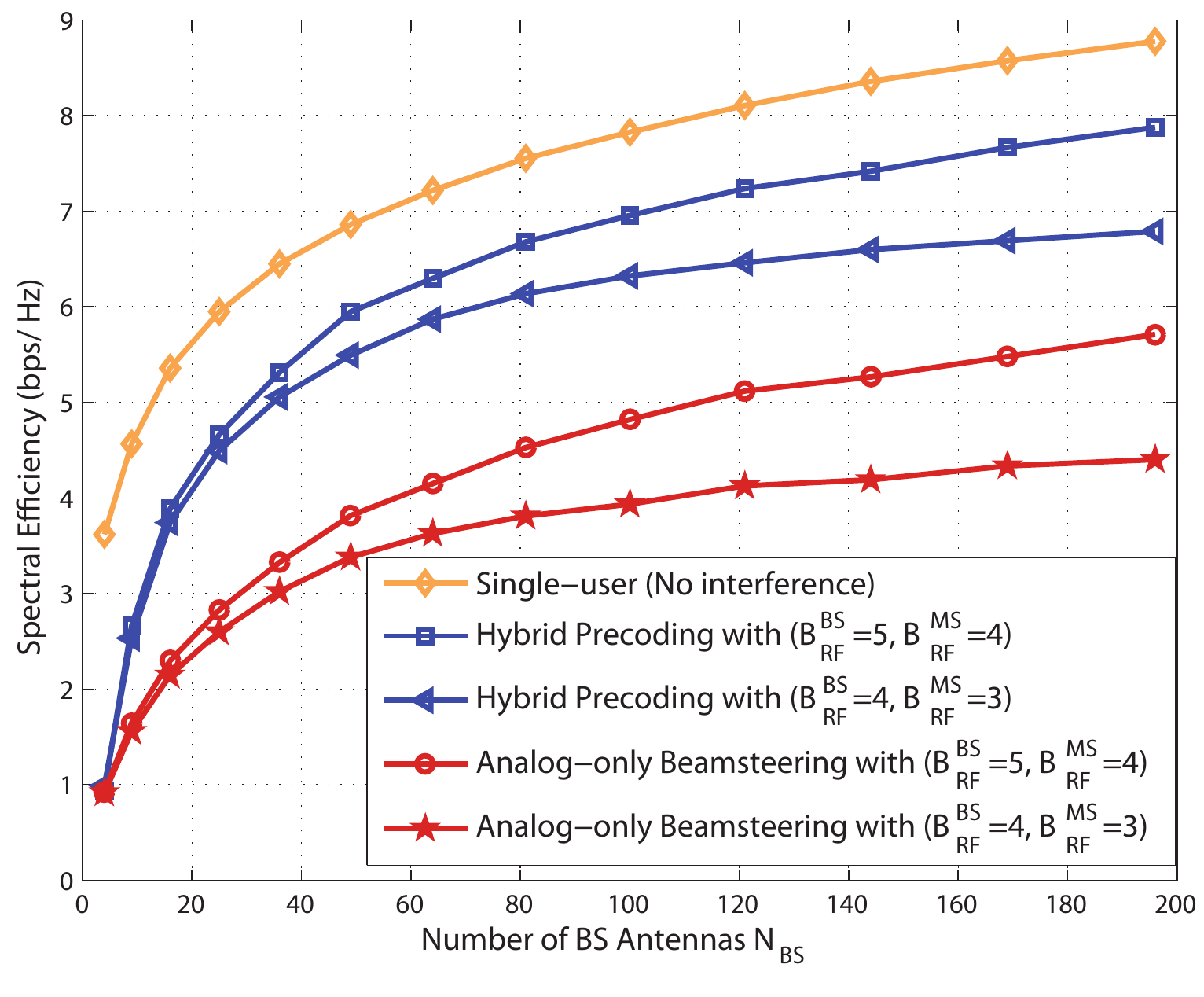}
\label{fig:AQ}}
\subfigure[center][{}]{
\includegraphics[width=0.47\columnwidth, height=180pt]{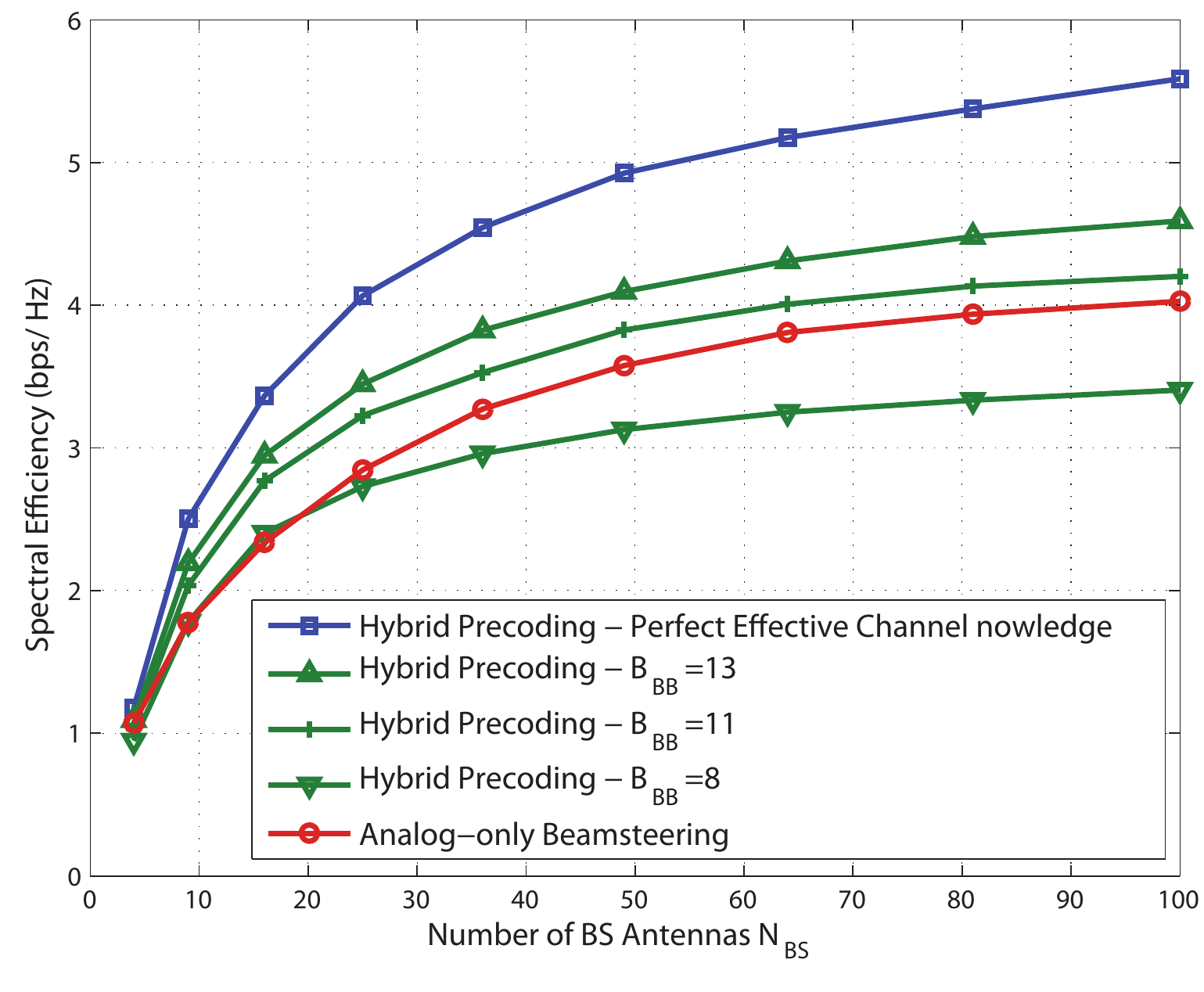}
\label{fig:RVQ}}
\caption{Achievable rates using the hybrid precoding and beamsteering algorithms are plotted for different numbers of RF beamforming quantization bits in (a), and for different numbers of effective channel quantization bits in (b). }
\label{fig:Fig5}
\vspace{-10pt}
\end{figure}

In \figref{fig:RVQ}, the case when both RF and baseband quantized codebooks exist is illustrated. For this  figure, the same system setup of \figref{fig:AQ} is adopted again, and the spectral efficiency achieved by hybrid precoding is shown for different sizes of the RVQ codebook used in quantizing the effective channels. The RF codebooks are also quantized with $B_\mathrm{RF}^\mathrm{BS}=3$ bits and $B_\mathrm{RF}^\mathrm{MS}=2$ bits. These results show  that when the effective channel is poorly quantized, the loss of multi-user interference management is larger than its gain, and using analog-only beamsteering achieves better rates. For reasonable numbers of effective channel quantization bits, however, the performance of hybrid precoding maintains its gain over the described analog-only solutions.

Finally, \figref{fig:Coverage} evaluates the performance of the proposed hybrid precoding algorithm in a mmWave cellular setup including inter-cell interference, which is not explicitly incorporated into our designs. In this setup, BS's and MS's are assumed to be spatially distributed according to a Poisson point process with MS's densities 30 times the BS densities. The channels between the BS's and MS's are  single-path and each link is determined to be line-of-sight or non-line-of-sight based on the blockage model in \cite{Cov_Magazine}. Each MS is associated to the BS with less path-loss and the BS randomly selects $ n=2,..,5$ users of those associated to it to be simultaneously served. BS's are assumed to have $8 \times 8$ UPAs and MS's are equipped with $4 \times 4$ UPAs. All UPA's are vertical, elevation angles are assumed to be fixed at $\pi/2$, and azimuth angles are uniformly distributed in $[0, 2 \pi]$. \figref{fig:Coverage} shows the per-user coverage probability defined as $\mathcal{P}\left(\mathrm{R_u \geq \eta}\right)$, where $\eta$ is an arbitrary threshold. This figure illustrates that hybrid precoding has a reasonable coverage gain over analog-only beamsteering, especially when large numbers of users are simultaneously served, thanks to the interference management capability of hybrid precoding.

\begin{figure}[t]
\centerline{
\includegraphics[scale=.49]{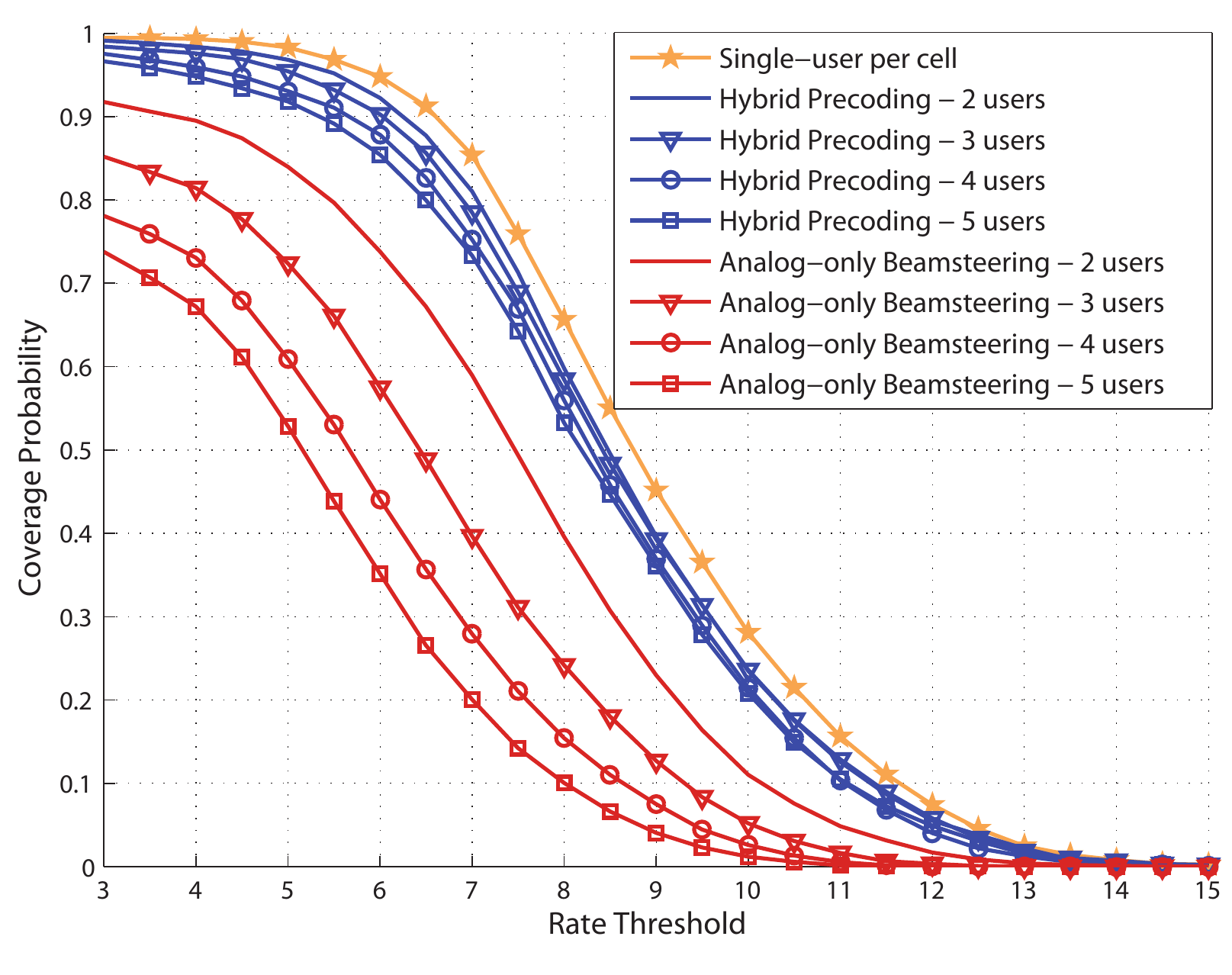}
}
\caption{Coverage probability of the proposed hybrid precoding algorithm compared with single-user per cell and analog-only beamsteering solutions. The figure shows the per-user performance with different numbers of users per cell.}
\label{fig:Coverage}
\vspace{-10pt}
\label{fig:Fig6}
\end{figure}

\section{Conclusions}\label{sec:conclusion}
In this paper, we proposed a low-complexity hybrid analog/digital precoding algorithm for downlink multi-user mmWave systems leveraging the sparse nature of the channel and the large number of deployed antennas. The performance of the proposed algorithm was analyzed when the channels are single-path and when the system dimensions are very large. In these cases, the asymptotic optimality of the proposed algorithm, and the gain over beamsteering solutions were illustrated. The results indicate that interference management in multi-user mmWave systems is required even when the number of antennas is very large. When the feedback channels are limited, the average rate loss due to joint analog/digital codebook quantization was analyzed and numerically simulated. These simulations show that the hybrid precoding gain is not very sensitive to RF angles quantization. It is important, however, to have a good quantization for the digital precoding layer to maintain a reasonable precoding gain over analog only solutions. As a future work, it is of interest to develop efficient mmWave precoding and channel estimation algorithms for multi-user cellular systems taking into consideration the out-of-cell interference.
\appendices
\section{}\label{app:power}
\textit{Power constraints of the digital precoder}: To satisfy the hybrid precoding power constraints, we need to adjust the diagonal elements in $\boldsymbol\Lambda$ such that $\left\|\bF_\mathrm{RF}\bff_u^\mathrm{BB}\right\|_2^2=1, u=1, 2, ..., U$. If we set ${\boldsymbol\Lambda}_{u,u}=\sqrt{\frac{N_\mathrm{BS} N_\mathrm{MS} }{ \left(\bA_\mathrm{BS}^* \bA_\mathrm{BS}\right)^{-1}_{u,u}}}\left|\alpha_u\right|$, and use the effective channel definition in \eqref{eq:eff_channel}, we get

\small
\begin{align}
\left\|\bF_\mathrm{RF}\bff_u^\mathrm{BB}\right\|_2^2&=\left[\boldsymbol\Lambda\right]_{u,:}\left(\overline{\bH}_u \overline{\bH}_u^*\right)^{-1} \overline{\bH}_u \bF_\mathrm{RF}^* \bF_\mathrm{RF}\overline{\bH}_u^* \left(\overline{\bH}_u \overline{\bH}_u^*\right)^{-1} \left[\boldsymbol\Lambda\right]_{:,u}\\
&\stackrel{}{=}\left[\boldsymbol\Lambda\right]_{u,:}\left(\bD^*\right)^{-1} \left( \bA_\mathrm{BS}^* \bA_\mathrm{BS} \bA_\mathrm{BS}^* \bA_\mathrm{BS} \right)^{-1} \bA_\mathrm{BS}^* \bA_\mathrm{BS} \bA_\mathrm{BS}^* \bA_\mathrm{BS} \bA_\mathrm{BS}^* \bA_\mathrm{BS} \\
& \hspace{40pt} \times \left(\bA_\mathrm{BS}^* \bA_\mathrm{BS} \bA_\mathrm{BS}^* \bA_\mathrm{BS} \right)^{-1} \left(\bD\right)^{-1} \left[\boldsymbol\Lambda\right]_{:,u}\\
\hspace{10pt}&=\left[\boldsymbol\Lambda\right]_{u,:}\left(\bD^*\right)^{-1}  \left( \bA_\mathrm{BS}^* \bA_\mathrm{BS} \right)^{-1} \left(\bD\right)^{-1} \left[\boldsymbol\Lambda\right]_{:,1}\\
&=\left[\boldsymbol\Lambda\right]_{u,u}^2 \left[\bD\right]^{-2}_{u,u}  \left( \bA_\mathrm{BS}^* \bA_\mathrm{BS} \right)^{-1}_{u,u} \\
&=1.
\end{align}
\section{}\label{app:Cor1}
\normalsize
\begin{proof}[Proof of Proposition \ref{Prop1}]  The difference between the average rate of user $u$ and the single-user rate  can be written as
\small
\begin{align}
\bbE\left[ \mathring{R}-R_u\right]&\stackrel{(a)}{\leq} \bbE\left[\log2\left(1+ \frac{SNR}{  U}N_\mathrm{BS} N_\mathrm{MS} \left|\alpha_u\right|^2\right)-\log_2\left( 1+\frac{SNR}{ U}{N_\mathrm{BS} N_\mathrm{MS} \left|\alpha_u\right|^2} G\left(\left\{\phi_u\right\}_{u=1}^U\right) \right)\right]\\
& \leq \bbE\left[- \log_2\left(G\left(\left\{\phi_u\right\}_{u=1}^U\right)\right) \right] \\
& \stackrel{(b)}{=} K\left(N_\mathrm{BS},U\right)
\end{align}
\normalsize
where (a) follows from the lower bound of the achievable rate in Theorem \ref{Th1}, and (b) is due to the definition of $G\left(\left\{\phi_u\right\}_{u=1}^U\right)$ which depends only on the steering vectors, and the user AoDs.

To prove the second part of Proposition \ref{Prop1}, denote $\beta_{u,n}=\ba^*_\mathrm{BS}\left(\phi_u\right)\ba_\mathrm{BS}\left(\phi_n\right)$. Then, we note that the off-diagonal elements of the matrix $\bA_\mathrm{BS}^* \bA_\mathrm{BS}$, i.e., $\beta_{u,n}, n \neq u$ satisfy lemma 1 in \cite{ElAyache2} which leads to:
$\lim_{N_\mathrm{BS}\rightarrow \infty} \beta_{u,n}=0$ with probability one, $\forall n \neq u$.
Consequently, we have $\lim_{N_\mathrm{BS}\rightarrow \infty} \bA_\mathrm{BS}^* \bA_\mathrm{BS}=\bI_U$ with probability one. Hence, the matrix $\bA_\mathrm{BS}$ becomes semi-unitary with $\sigma_\mathrm{max}\left(\bA_\mathrm{BS}\right)=\sigma_\mathrm{min}\left(\bA_\mathrm{BS}\right)=1$ which results in $G\left(\left\{\phi_u\right\}_{u=1}^U\right)=1$ and $R_u=\mathring{R}$ with probability one.
\end{proof}
\section{}\label{app:Prop2}
\begin{proof}[Proof of Corollary \ref{Prop2}]
The gain of the achievable rate using Algorithm \ref{alg:MU_Precoding} over the rate, $R_\mathrm{BS}$, achieved by having only an RF beamsteering precoder $\bF_\mathrm{RF}=\bA_\mathrm{BS}$ is bounded by
\small
\begin{align}
R_u-R_\mathrm{BS} &\stackrel{(a)}{\geq} \log_2\left(\frac{SNR}{ U} N_\mathrm{BS} N_\mathrm{MS} \left|\alpha_u\right|^2 G\left(\left\{\phi_u\right\}_{u=1}^U\right) \right)-\log_2\left(1+\frac{\frac{SNR}{ U}  N_\mathrm{BS} N_\mathrm{MS} \left|\alpha_u\right|^2 }{\frac{SNR}{ U}  N_\mathrm{BS} N_\mathrm{MS} \left|\alpha_u\right|^2 \sum_{n \neq u} \left|\beta_{u,n}\right|^2+1}\right)\\
&=\log_2\left(N_\mathrm{MS}\right)+\log_2\left(\frac{SNR}{ U} N_\mathrm{BS}\left|\alpha_u\right|^2 G\left(\left\{\phi_u\right\}_{u=1}^U\right) \right)\\
&-\log_2\left(1+\frac{\frac{SNR}{ U}  N_\mathrm{BS}  \left|\alpha_u\right|^2 }{\frac{SNR}{ U}  N_\mathrm{BS}  \left|\alpha_u\right|^2 \sum_{n \neq u} \left|\beta_{u,n}\right|^2+\frac{1}{N_\mathrm{MS}}}\right),
\end{align}
\normalsize
where (a) follows from the hybrid precoding rate lower bound in Theorem \ref{Th1}, and the rate using the beamsteering alone. Now, taking the expectation over the channel gain  and AoDs, and setting $N_\mathrm{MS}\rightarrow \infty$, we get the first term goes to infinity and the others terms become constant.
\end{proof}
\section{}\label{app:Prop_Finite}
\begin{proof}[Proof of Theorem \ref{Th2}]
Consider the system model described in \sref{sec:Model}, with the hybrid analog/digital precoders and RF combiners designed using Algorithm \ref{alg:MU_Precoding}. In the first stage, the BS and each MS $u$ selects $\bv_u^\star$ and $\bg_u^\star$  that solve \eqref{eq:Opt}.

Given the quantized codebooks $\cF, \cW$ and the single-path channels, $\bg_u^\star=\ba_\mathrm{MS}\left(\hat{\theta}_u\right)$, where $\hat{\theta}_u$ is the quantized angle in the MS beamsteering codebook $\cW$ that maximizes $\left|\ba^*_\mathrm{MS}\left(\hat{\theta}_u\right) \ba_\mathrm{MS}\left(\theta_u\right)\right|$. Similarly, $\bv_u^\star=\ba_\mathrm{BS}\left(\hat{\phi_u}\right)$ where $\hat{\phi}_u$ is the quantized angle in the BS beamsteering codebook $\cF$ that maximizes $\left|\ba^*_\mathrm{BS}\left(\hat{\phi}_u\right) \ba_\mathrm{MS}\left(\phi_u\right)\right|$. Finally, define $\mu_\mathrm{MS}\left(\hat{\theta}_u,\theta_u\right)=\ba^*_\mathrm{MS}\left(\hat{\theta}_u\right) \ba_\mathrm{MS}\left(\theta_u\right)$ and $\mu_\mathrm{BS}\left(\hat{\phi}_u,\phi_u\right)=\ba^*_\mathrm{BS}\left(\hat{\phi}_u\right) \ba_\mathrm{BS}\left(\phi_u\right)$.

The BS can then construct its RF precoding matrix $\bF_\mathrm{RF}=\left[\ba_\mathrm{BS}\left(\hat{\theta}_1\right), \ba_\mathrm{BS}\left(\hat{\theta}_2\right), ..., \ba_\mathrm{BS}\left(\hat{\theta}_U\right)\right]$, and the effective channel of user $u$ is written as
\small
\begin{equation}
\overline{\bh}_u^\mathrm{Q}=\sqrt{N_\mathrm{BS} N_\mathrm{MS}} \alpha_u \mu_\mathrm{MS} \left(\hat{\theta}_u,\theta_u\right) \mu_\mathrm{BS}^*\left(\hat{\phi}_u,\phi_u\right) \left[\beta_{u,1}^\mathrm{Q}, \beta_{u,2}^\mathrm{Q}. ..., \beta_{u,U}^\mathrm{Q}\right],
\label{eq:eff_q_channel}
\end{equation}
\normalsize
where
\small
\begin{equation}
\beta_{u,n}^\mathrm{Q}=\frac{\ba^*_\mathrm{BS}\left(\phi_u\right) \ba_\mathrm{BS}\left(\hat{\phi}_n\right)}{\mu_\mathrm{BS}^*\left(\hat{\phi}_u,\phi_u\right)}.
\label{eq:qu_beta}
\end{equation}
\normalsize

In the second stage of Algorithm \ref{alg:MU_Precoding}, each MS gets perfect knowledge of its effective channel $\overline{\bh}_u^\mathrm{Q}$, which includes the effect of the quantized beamsteering directions. Next, this MS quantizes its normalized effective channel $\widetilde{\bh}_u^\mathrm{Q}=\frac{\overline{\bh}_u^\mathrm{Q}}{\left\|\overline{\bh}_u^\mathrm{Q}\right\|}$ using the RVQ codebook $\cH$, and selects $\widehat{\bh}_u^\mathrm{Q}$ that solves: $\widehat{\bh}_u^\mathrm{Q}=\arg\max_{\bg \in \cH} \left|{\widetilde{\bh}_u^{\mathrm{Q}^*} } \bg \right|$.

\normalsize
Based on this quantized effective channel, the BS builds its digital zero-forcing precoder, and normalizes each column of it, $\widehat{\bff}_u^\mathrm{BB}$, similar to \sref{subsec:single_path}.

If $R_u^\mathrm{Q}$ denotes the resulting rate of user $u$, then the average rate loss compared with the rate without RF and baseband quantization, $\overline{\Delta R}_u$, can be written as
\small
\begin{align}
\overline{\Delta R}_u&=\bbE\left[\log_2\left(1+\frac{\mathsf{SNR}}{U}\left|\overline{\bh}_u^* \bff_u^\mathrm{BB}\right|^2\right)-\log2\left(1+\frac{\frac{\mathsf{SNR}}{U}\left|\overline{\bh}_u^{\mathrm{Q}^*} \widehat{\bff}_u^\mathrm{BB}\right|^2}{\frac{\mathsf{SNR}}{U} \sum_{n \neq u}^U \left|\overline{\bh}_u^{\mathrm{Q}^*} \widehat{\bff}_n^\mathrm{BB}\right|^2 +1}\right)\right], \\
&\begin{aligned}&=\bbE\left[\log_2\left(1+\frac{\mathsf{SNR}}{U}\left|\overline{\bh}_u^* \bff_u^\mathrm{BB}\right|^2\right)-\log2\left(1+\frac{\mathsf{SNR}}{U} \sum_{n \neq u}^U \left|\overline{\bh}_u^{\mathrm{Q}^*} \widehat{\bff}_n^\mathrm{BB}\right|^2 + \frac{\mathsf{SNR}}{U}\left|\overline{\bh}_u^{\mathrm{Q}^*} \widehat{\bff}_u^\mathrm{BB}\right|^2 \right)\right]\\
&+\bbE\left[\log_2\left(1+\frac{\mathsf{SNR}}{U} \sum_{n \neq u}^U \left|\overline{\bh}_u^{\mathrm{Q}^*} \widehat{\bff}_n^\mathrm{BB}\right|^2 \right)\right],
\end{aligned}\\
&\begin{aligned}&\stackrel{(a)}{\leq}\bbE\left[\log_2\left(\frac{1+\frac{\mathsf{SNR}}{U}\left|\overline{\bh}_u^* \bff_u^\mathrm{BB}\right|^2}{1+ \frac{\mathsf{SNR}}{U}\left|\overline{\bh}_u^{\mathrm{Q}^*} \widehat{\bff}_u^\mathrm{BB}\right|^2} \right)\right]
+\bbE\left[\log_2\left(1+\frac{\mathsf{SNR}}{U} \sum_{n \neq u}^U \left|\overline{\bh}_u^{\mathrm{Q}^*} \widehat{\bff}_n^\mathrm{BB}\right|^2 \right)\right],
\end{aligned}\\
&\begin{aligned}&\stackrel{(b)}{\leq}\bbE\left[\log_2\left(\frac{\left\|\overline{\bh}_u \right\|^2}{\left\|\overline{\bh}_u^{\mathrm{Q}}\right\|^2}\right)\right]+\bbE\left[\log_2\left(\frac{\mathsf{SNR}}{U} \left|\widetilde{\bh}_u^* \bff_u^\mathrm{BB}\right|^2 \right)\right] \\
& - \bbE\left[\log_2\left(\frac{\mathsf{SNR}}{U}\left|\widetilde{\bh}_u^{\mathrm{Q}^*} \widehat{\bff}_u^\mathrm{BB}\right|^2 \right)\right] +\bbE\left[\log_2\left(1+\frac{\mathsf{SNR}}{U} \sum_{n \neq u}^U \left|\overline{\bh}_u^{\mathrm{Q}^*} \widehat{\bff}_n^\mathrm{BB}\right|^2 \right)\right],
\end{aligned}
\end{align}
\normalsize
where (a) resulted from removing the positive quantity $\frac{\mathsf{SNR}}{U} \sum_{n \neq u}^U \left|\overline{\bh}_u^{\mathrm{Q}^*} \widehat{\bff}_n^\mathrm{BB}\right|^2$ from the second term, and (b) is by noting that  any positive numbers $x, y$, with $x>y$, satisfy $\log(1+x)-\log(1+y) \leq \log(x)-\log(y)$. Next, as the zero-forcing baseband precoding vectors $\bff_u^\mathrm{BB},\widehat{\bff}_u^\mathrm{BB}$ are designed to be in the null space of the other users' channel vectors, we have $\bff_u^\mathrm{BB}$, and $\widehat{\bff}_u^\mathrm{BB}$ independent from $\widetilde{\bh}_u$, and $\widetilde{\bh}_u^\mathrm{Q}$ respectively, and the expectations of their projections are equal which yields \cite{Jindal, Raghavan}.
\small
\begin{align}
\overline{\Delta R}_u&{=}\bbE\left[\log_2\left(\frac{1+\sum_{n \neq u}{\left|\beta_{u,n}\right|^2}}{\left|\mu_\mathrm{MS}\left(\hat{\theta}_u,\theta_u\right)\right|^2\left|\mu_\mathrm{BS}\left(\hat{\phi}_u,\phi_u\right)\right|^2 \left(1+\sum_{n \neq u}{\left|\beta_{u,n}^\mathrm{Q}\right|^2}\right)}\right)
+\log_2\left(1+\frac{\mathsf{SNR}}{U} \sum_{n \neq u}^U \left|\overline{\bh}_u^{\mathrm{Q}^*} \widehat{\bff}_n^\mathrm{BB}\right|^2 \right)\right],\\
&\begin{aligned}&\stackrel{(a)}{\leq}\bbE\left[\log_2\left(\frac{1}{\left|\mu_\mathrm{MS}\left(\hat{\theta}_u,\theta_u\right)\right|^2\left|\mu_\mathrm{BS}\left(\hat{\phi}_u,\phi_u\right)\right|^2  }\right)\right] +\bbE\left[\log_2\left(1+\frac{\mathsf{SNR}}{U} \sum_{n \neq u}^U \left|\overline{\bh}_u^{\mathrm{Q}^*} \widehat{\bff}_n^\mathrm{BB}\right|^2 \right)\right],
\end{aligned}
\end{align}
\normalsize
where (a) follows from the definition of $\beta_{u,n}^\mathrm{Q}$ in \eqref{eq:qu_beta} which means that $\bbE\left[\log_2\left(1+\sum_{n \neq u}{\left|\beta_{u,n}\right|^2}\right)\right] \leq \bbE\left[\log_2\left(1+\sum_{n \neq u}{\left|\beta_{u,n}^\mathrm{Q}\right|^2}\right)\right]$.

Now, we define the bounds $\left|\mu_\mathrm{MS}\left(\hat{\theta}_u,\theta_u\right)\right| \geq \left|\overline{\mu}_\mathrm{BS}\right|=\displaystyle{\min_{\bff_u \in \cF} \max_{\bff_n \in \cF}} {\left| \bff_u^* \bff_n\right|}$ and $\left|\mu_\mathrm{BS}\left(\hat{\phi}_u,\phi_u\right)\right| \geq \left|\overline{\mu}_\mathrm{MS}\right|=\displaystyle{\min_{\bw_u \in \cW} \max_{\bw_n \in \cW}} {\left| \bw_u^* \bw_n\right|}$. Using these bounds and applying Jensen's inequality, we get
\small
\begin{align}
\overline{\Delta R}_u & \leq \log_2\left(\frac{1}{\left|\overline{\mu}_\mathrm{MS}\right|^2\left|\overline{\mu}_\mathrm{BS}\right|^2  }\right) +\log_2\left(1+\frac{\mathsf{SNR}}{U} (U-1) \bbE\left[\left\|\overline{\bh}_u^{\mathrm{Q}^*}\right\|^2\right] \bbE\left[\left|\widetilde{\bh}_u^{\mathrm{Q}^*} \widehat{\bff}_n^\mathrm{BB}\right|^2\right] \right), \\
&\begin{aligned} &\stackrel{(a)}{\leq} \log_2\left(\frac{1}{\left|\overline{\mu}_\mathrm{MS}\right|^2\left|{\mu}_\mathrm{BS}\right|^2  }\right) \\
&+\log_2\left(1+\frac{\mathsf{SNR}}{ U} (U-1) N_\mathrm{BS}N_\mathrm{MS} \left|\mu_\mathrm{MS}\left(\hat{\theta}_u,\theta_u\right)\right|^2 \bar{\alpha} \left(\left|\mu_\mathrm{BS}\left(\hat{\phi}_u,\phi_u\right)\right|^2+\frac{U-1}{N_\mathrm{BS} } \right)\bbE\left[\left|\widetilde{\bh}_u^{\mathrm{Q}^*} \widehat{\bff}_n^\mathrm{BB}\right|^2\right] \right),
\end{aligned} \\
&\begin{aligned} &\stackrel{(b)}{\leq} \log_2\left(\frac{1}{\left|\overline{\mu}_\mathrm{MS}\right|^2\left|\overline{\mu}_\mathrm{BS}\right|^2  }\right)
+\log_2\left(1+\frac{\mathsf{SNR}}{ U} (U-1) N_\mathrm{BS}N_\mathrm{MS} \bar{\alpha} \left(1+\frac{U-1}{N_\mathrm{BS}} \right)\bbE\left[\left|\widetilde{\bh}_u^{\mathrm{Q}^*} \widehat{\bff}_n^\mathrm{BB}\right|^2\right] \right),
\end{aligned} \\
&\begin{aligned} &\stackrel{(c)}{=} \log_2\left(\frac{1}{\left|\overline{\mu}_\mathrm{MS}\right|^2\left|\overline{\mu}_\mathrm{BS}\right|^2  }\right) +\log_2\left(1+\frac{\mathsf{SNR}}{U} N_\mathrm{BS}N_\mathrm{MS}  \bar{\alpha} \left(1+\frac{U-1}{N_\mathrm{BS} } \right)\bbE\left[ \arg\min_{\substack{\bg \in \cH \\ \left|\cH\right|=2^{B_\mathrm{BB}}}} \sin^2\left(\widetilde{\bh}_u^\mathrm{Q},\bg\right) \right]\right),
\end{aligned} \\
&\begin{aligned} \stackrel{(d)}{\leq}\log_2\left(\frac{1+\frac{\mathsf{SNR}}{ U} N_\mathrm{BS} N_\mathrm{MS} \bar{\alpha} \left(1+\frac{U-1}{N_\mathrm{BS} }\right) 2^{-\frac{B_\mathrm{BS}}{U-1}}}{\left|\overline{\mu}_\mathrm{BS}\right|^2 \left|\overline{\mu}_\mathrm{MS}\right|^2}\right),
\end{aligned}
\end{align}
\normalsize
where (a) is by taking the expectation of the effective channel in \eqref{eq:eff_q_channel}, and using the lower bound on $\left|\mu_\mathrm{MS}\left(\hat{\theta}_u,\theta_u\right)\right|$ and $\left|\mu_\mathrm{BS}\left(\hat{\phi}_u,\phi_u\right)\right|$. (b) is by using the upper bound on $\left|\mu_\mathrm{MS}\left(\hat{\theta}_u,\theta_u\right)\right| \leq 1$ and $\left|\mu_\mathrm{BS}\left(\hat{\phi}_u,\phi_u\right)\right| \leq 1$. Now, using a similar trick as in \cite{Jindal,Raghavan}, we can write $\widetilde{\bh}_u^\mathrm{Q}=\sqrt{1-a}\widehat{\bh}_u+a \bz$, with $\bz$ a unit-norm vector in the null-space of $\widehat{\bh}_u$, and $a=\sin^2\left(\widetilde{\bh}_u^\mathrm{Q},\widehat{\bh}_u\right)$. Exploiting the orthogonality between $\widehat{\bff}_n^\mathrm{BB}$, and $\widehat{\bh}_u$, we have $\left|\widetilde{\bh}_u^{\mathrm{Q}^*} \widehat{\bff}_n^\mathrm{BB}\right|^2=a \left|\bz^* \widehat{\bff}_n^\mathrm{BB}\right|^2$. As both $\bz$ and $\widehat{\bff}_n^\mathrm{BB}$ are independent and isotropically independently in the $(U-1)$-dimensional null-space of $\widehat{\bh}_u$, we get $\bbE\left[\left|\bz^* \widehat{\bff}_n^\mathrm{BB}\right|^2\right]=\frac{1}{U-1}$ which results in (c). Finally, as the effective channel in \eqref{eq:eff_q_channel} is distributed only in a sub-space of the K-dimensional space in which the codewords in $\cH$ are uniformly distributed, the average quantization error can not be greater than the case when the channel vector is distributed in the entire space, which is upper bounded by $2^{-\frac{B_\mathrm{BB}}{U-1}}$ \cite{Jindal}.
\end{proof}

\linespread{1.3}

\end{document}

%% file: input.tex
\usepackage{amsfonts}
\usepackage{times}
\usepackage{graphicx}
\usepackage{latexsym}
\usepackage{dsfont}
\usepackage{amssymb}
\usepackage{amsmath}
\usepackage{cite}
\usepackage{verbatim}
\usepackage{subfigure}
\newtheorem{theorem}{Theorem}

\newtheorem{algorithm}[theorem]{Algorithm}

\newtheorem{corollary}[theorem]{Corollary}

\newtheorem{lemma}[theorem]{Lemma}

\newtheorem{proposition}[theorem]{Proposition}

\newenvironment{proof}{ \textbf{Proof:} }{ \hfill $\Box$}

\newcommand{\figref}[1]{{Fig.}~\ref{#1}}

\def\bb0{{\mathbb{0}}}

\def\ba{{\mathbf{a}}}
\def\bb{{\mathbf{b}}}

\def\bff{{\mathbf{f}}}
\def\bg{{\mathbf{g}}}
\def\bh{{\mathbf{h}}}

\def\bm{{\mathbf{m}}}
\def\bn{{\mathbf{n}}}

\def\br{{\mathbf{r}}}
\def\bs{{\mathbf{s}}}

\def\bv{{\mathbf{v}}}
\def\bw{{\mathbf{w}}}
\def\bx{{\mathbf{x}}}

\def\bz{{\mathbf{z}}}
\def\b0{{\mathbf{0}}}

\def\bA{{\mathbf{A}}}

\def\bD{{\mathbf{D}}}

\def\bF{{\mathbf{F}}}

\def\bH{{\mathbf{H}}}
\def\bI{{\mathbf{I}}}

\def\bP{{\mathbf{P}}}

\def\bR{{\mathbf{R}}}

\def\bbE{{\mathbb{E}}}

\def\bbP{{\mathbb{P}}}

\def\cA{\mathcal{A}}

\def\cC{\mathcal{C}}

\def\cF{\mathcal{F}}

\def\cH{\mathcal{H}}

\def\cN{\mathcal{N}}

\def\cW{\mathcal{W}}

\def\sf0{{\mathsf{0}}}

